\documentclass[english,11pt]{article}

\usepackage{geometry}
\usepackage[sort,comma]{natbib}
\usepackage[pagebackref=false,colorlinks,citecolor=blue,urlcolor=blue,pdffitwindow=true,linkcolor=blue]{hyperref}
\usepackage{amsmath}
\usepackage{bbm}
\usepackage{multirow}
\usepackage{graphicx}
\DeclareGraphicsExtensions{.pdf,.jpg,.png}
\usepackage{subfigure}

\usepackage{amsfonts}
\usepackage{mathdots}

\usepackage{amsthm}
\newtheorem{remark}{Remark}

\usepackage{multirow}

\usepackage{centernot}

\usepackage[font={footnotesize,it}, sc, bf]{caption}
\renewenvironment{abstract}{%
  \noindent\textbf\abstractname .\hspace{1pt}
}{%
  \endlist \par\bigskip\bigskip
}

\usepackage{lastpage}
\RequirePackage[hyperpageref]{backref}
\renewcommand*{\backref}[1]{} 
\renewcommand*{\backrefalt}[4]{
    \ifcase #1
       No referred.
    \or
       \emph{Referred to on page #2.}
    \else
       \emph{Referred to on pages #2.}
    \fi}

\newtheorem{lemma}{Lemma}
\newtheorem{proposition}{Proposition}
\newtheorem{assumption}{Assumption}
\newtheorem{corollary}{Corollary}

\usepackage[shortlabels]{enumitem}

\usepackage[titletoc,title]{appendix}

\binoppenalty=\maxdimen
\relpenalty=\maxdimen

\begin{document}

\title{\null\vspace{-60pt}{\bf Fast Bayesian inference in large Gaussian graphical models}\bigskip}
\date{}
\author{Gwena\"el G.R. Leday and Sylvia Richardson \\[5pt]
\emph{MRC Biostatistics Unit, University of Cambridge, UK}}

{\maketitle}
\begin{center}

\end{center}
\bigskip\bigskip

\begin{abstract}
Despite major methodological developments, Bayesian inference for Gaussian graphical models remains challenging in high dimension due to the tremendous size of the model space. This article proposes a method to infer the marginal and conditional independence structures between variables by multiple testing of hypotheses. Specifically, we introduce closed-form Bayes factors under the Gaussian conjugate model to evaluate the null hypotheses of marginal and conditional independence between variables. Their computation for all pairs of variables is shown to be extremely efficient, thereby allowing us to address large problems with thousands of nodes. Moreover, we derive exact tail probabilities from the null distributions of the Bayes factors. These allow the use of any multiplicity correction procedure to control error rates for incorrect edge inclusion. We demonstrate the proposed approach to graphical model selection on various simulated examples as well as on a large gene expression data set from The Cancer Genome Atlas.
\end{abstract}

\section{Introduction}

Graphical models provide a natural basis for the statistical description and analysis of interplay between variables. In applications, interest often lies in the bidirected and undirected graphs that respectively describe the marginal and conditional dependence structures among variables \citep{wang2015}. When the joint distribution of the variables is assumed to be Gaussian, these are known to be fully coded in the covariance matrix $\Sigma = \{ \sigma_{ij} \}$ and its inverse $\Omega = \{ \omega_{ij} \}$ \citep{dempster1972, cox1993}. Precisely, a pair $(i,j)$ of variables with $1\leq i<j\leq p$, will be marginally independent when $\sigma_{ij} = 0$ and conditionally independent (given all the remaining variables) when $\omega_{ij} = 0$. The present article treats inference of both types of graphs in context of the Gaussian model when the number of variables $p$ is large.

Despite major methodological developments, Bayesian inference for Gaussian graphical models remains challenging. The standard approach casts the problem as a model selection problem, and first requires specification of prior distributions over all possible graphical models and their parameter spaces. Such specification is not straightforward as it is desirable to favour parsimonious models and address the compatibility of priors across models \citep{carvalho2009, consonni2012}. Next, the inference procedure is hindered by the search over a very high-dimensional model space where the number of possible graphical models grows super-exponentially with the number of variables. Full exploration of the model space is, therefore, only possible when the number of variables is very small (say $p\leq 10$). In moderate- and high-dimensional settings where $p$ is in the tens, hundreds or thousands, the model space must generally be searched stochastically \citep{mohammadi2015,wang2012,lenkoski2011,giudici1999}. However, due to the tremendous size of the model space in such settings, it may be difficult (nay impossible) to identify with confidence the graphical model that is best supported by the data. Accordingly, it has become common practice to account for model uncertainty by performing Bayesian model averaging and infer the graphical structure by selecting edges with the highest marginal posterior probabilities, for example by exploiting their connection to a Bayesian version of the false discovery rate \citep{mitra2013, peterson2015, baladandayuthapani2014}.

To tackle the difficulties associated with the standard approach this article proposes a method to directly select edges by multiple testing of hypotheses about pairwise (marginal or conditional) independence \citep{drton2007} using closed-form Bayes factors. These are obtained using the conditional approach of \citet{dickey1971}, in which the prior under the null hypothesis is derived from that of the alternative by conditioning on the null hypothesis. This approach was also adopted by \citet{giudici1995} to derive a closed-form Bayes factor for conditional independence. However, the latter relies on elements of the inverse of the sample covariance matrix which is singular when the number of variables is large relative to the sample size. We bypass this issue and also introduce a closed-form Bayes factor for marginal independence. Moreover, we show the consistency of the Bayes factors and derive exact tail probabilities from their null distributions to help address the multiplicity problem and control error rates for incorrect edge inclusion. The proposed procedure, available via the R package beam on the CRAN website \citep[http://cran.r-project.org]{R2016}, is shown to be computationally very efficient, addressing problems with thousands of nodes in just a few seconds.

The article is structured as follows. The next section introduces basic notations and the Gaussian conjugate model. In section~\ref{BF} we present closed-form Bayes factors to evaluate the null hypotheses of marginal and conditional independence between any two variables and study their consistency. Section~\ref{multest} details graph inference and discuss the multiple testing problem and error control. The performance of the proposed approach is compared to Bayesian and non-Bayesian methods on simulated data in section~\ref{num}. Section~\ref{app} illustrates our method on a large gene expression data set from The Cancer Genome Atlas.

\newpage
\section{Background}
\label{bg}

\subsection{Notation}
\label{bg:not}

The following notation will be used throughout this paper.
We employ the notation $x \mid \mu, \Sigma \sim N_p (\mu, \Sigma)$ to say that the random column vector $x\in\mathbb{R}^p$,
has a multivariate normal distribution with mean $\mu\in\mathbb{R}^p$ and positive definite covariance matrix $\Sigma$.
We also write $\Omega \mid A, \alpha \sim W_d(A, \alpha)$ to indicate that the $d\times d$ random matrix $\Omega$ with density
\begin{equation*}
	(2)^{-\frac{\alpha d}{2}} \Gamma_d^{-1}\left(\frac{\alpha}{2}\right) \vert A \vert^{-\frac{\alpha}{2}} \vert \Omega \vert^{\frac{\alpha-d-1}{2}} \exp\left\lbrace -\frac{1}{2} \text{tr}( A^{-1} \Omega) \right\rbrace,
\end{equation*}
has a Wishart distribution with scale matrix $A$ and degrees of freedom $\alpha>d+1$. Here $\mid \Omega \mid$ represents the determinant of $\Omega$, $\text{tr}(A)$ denotes the trace of matrix $A$ and $\Gamma_d\left(x\right) = \pi^{\frac{d(d-1)}{4}}\prod_{i=1}^{d}{\Gamma\left( x + (1-i)/2 \right)}$ is the multivariate gamma function. The inverse of $\Omega$ is said to have an Inverse-Wishart distribution with scale matrix $A^{-1}$ and $\alpha$ degrees of freedom. We shall use the notation $\Omega^{-1} \mid A^{-1}, \alpha \sim IW_d(A^{-1}, \alpha)$. A random variable $\beta$ following a beta distribution with shape parameters $b_1$ and $b_2$ will be denoted by $\beta \sim Beta(b_1,b_2)$.
We use the operator $vec$ to denote the linear transformation that stacks the columns of a matrix into a vector and $\otimes$ to denote the Kronecker product. We refer to \citet{gupta2000} for more details on these operators.
Last, we shall use the subscripts $aa$, $bb$, $ab$ and $ba$ to refer to the submatrices $\Sigma_{aa}$, $\Sigma_{bb}$, $\Sigma_{ab}$ and $\Sigma_{ba}$ of a $p\times p$ symmetric matrix $\Sigma$ whose block-wise decomposition is implied by a partition of its rows and columns into two disjoint subsets indexed by $a\subset \{1,\ldots,p\}$ and $b=\{1,\ldots,p\} \setminus a$.

\subsection{The Gaussian conjugate model}
\label{bg:GCmodel}

Given an $n \times p$ observation matrix $Y = (Y_1, \ldots, Y_p)$ the Gaussian conjugate model is defined by
\begin{equation}
\label{GC}
	\begin{split}
		\text{vec}(Y) \mid \Sigma \sim N_{np} (0, \Sigma \otimes I_n)\quad\text{and}\quad \Sigma \mid D, \delta \sim  IW_p \left( \left(\delta-p-1\right) D,\delta \right)
	\end{split},
\end{equation}
with $D$ positive definite, $I_n$ the $n$-dimensional identity matrix and $\delta>p+1$. Here, the covariance matrix with kronecker product structure makes explicit the assumption of independence for the rows of $Y$ and the dependence of its columns via the covariance $\Sigma$.

Due to conjugacy, model~\eqref{GC} offers closed-form Bayesian estimators of the covariance matrix $\Sigma$ and its inverse $\Omega=\Sigma^{-1}$. The posterior expectation of $\Sigma$ is
\begin{equation}
\label{expectSigma}
		E\left(\Sigma \mid Y\right) = \left\lbrace\left(\delta-p-1\right) D+S\right\rbrace/(\delta+n-p-1),
\end{equation}
where $S=Y^T Y$, and that of its inverse is
\begin{equation}
\label{expectOmega}
	E\left(\Omega \mid Y\right) = (\delta+n) \left\lbrace\left(\delta-p-1\right) D+S\right\rbrace^{-1}.
\end{equation}

It is important to note that estimator~\eqref{expectSigma} is a linear shrinkage estimator that is a convex linear combination between the maximum likelihood estimator $\widehat{\Sigma}_{\text{mle}} = n^{-1}S$ of $\Sigma$ and $E\left(\Sigma\right) =  D$, with weight $\alpha = (\delta-p-1)/(\delta+n-p-1) \in (0,1)$ \citep{chen1979,hannart2014}. Likewise, estimator~\eqref{expectOmega} is recognized as a ridge-type estimator of the precision matrix \citep{kubokawa2008, vanWieringen2016}. The next proposition presents some properties of these two estimators. All proofs are presented in the Appendix.

\begin{proposition}
\label{prop1}
	Let estimators \eqref{expectSigma} and \eqref{expectOmega} depend on $\delta$ with $D$, $n$ and $p$ fixed, and denote them by $\widehat{\Sigma}_\delta$ and $\widehat{\Omega}_\delta$, respectively. Then the following properties hold:
	\begin{enumerate}[(i),leftmargin=*]
		\item $\lim_{\delta\rightarrow \infty} \widehat{\Sigma}_\delta = D$ 
		\item $\lim_{\delta\rightarrow \infty} \widehat{\Omega}_\delta = D^{-1}$
		\item $\lim_{\delta\rightarrow p+1} \widehat{\Sigma}_\delta = \widehat{\Sigma}_{\text{mle}}$ 
		\item $\lim_{\delta\rightarrow p+1} \widehat{\Omega}_\delta = \left\lbrace(n+p+1)/n\right\rbrace\widehat{\Sigma}_{\text{mle}}^{-1}$, if $n>p$,
		\item $\widehat{\Sigma}_\delta$ and $\widehat{\Omega}_\delta$ \text{are positive definite}
	\end{enumerate}
\end{proposition}

Additionally, the asymptotic properties of estimators~\eqref{expectSigma} and~\eqref{expectOmega} when $n\rightarrow\infty$ are the same as those of the maximum likelihood estimators $\widehat{\Sigma}_{\text{mle}}$ and $\widehat{\Sigma}_{\text{mle}}^{-1}$ of $\Sigma$ and $\Omega$. Proposition~\ref{prop2} summarizes.

\begin{proposition}
\label{prop2}
	Let estimator~\eqref{expectSigma} and~\eqref{expectOmega} depend on $n$ with $D$, $\delta$ and $p$ be fixed, and denote them by $\widehat{\Sigma}_{n}$ and $\widehat{\Omega}_{n}$, respectively. Then the following properties hold:
	\begin{enumerate}[(i),leftmargin=*]
		\item $\lim_{n\rightarrow\infty} \widehat{\Sigma}_{n} = \widehat{\Sigma}_{\text{mle}}$
		\item $\lim_{n\rightarrow\infty} \widehat{\Omega}_{n} = \widehat{\Sigma}_{\text{mle}}^{-1}$
	\end{enumerate}
\end{proposition}

\subsection{Choice of hyperparameters}
\label{bg:hyper}

In model~\ref{GC}, the prior matrix $D$ represents the prior expectation of $\Sigma$. It may also be interpreted as the shrinkage target towards which the maximum likelihood estimator of the covariance matrix is shrunk, since the posterior expectation of $\Sigma$ is a linear shrinkage estimator. For these reasons, $D$ can be chosen to encourage estimator~\eqref{expectSigma} to have specific structures (e.g. autoregressives or low-ranks). Ideally, in such cases the matrix $D$ should be parameterised by a low-dimensional vector of hyperparameters that are interpretable and for which prior knowledge exists. As often this knowledge is absent, it is common to choose $D=I_p$. Throughout this paper we use $D=I_p$ and standardize the $n\times p$ observation matrix $Y$ so that for $1\leq j \leq p$, $Y_{j}^T 1_n = 0$ and $Y_{j}^T Y_{j}/n = 1$, where $1_n$ is an $n\times 1$ vector whose elements are all equal to 1.

The other hyperparameter $\delta$ clearly acts as a regularization parameter (see equation~\eqref{expectSigma} and~\eqref{expectOmega}) and its value must therefore be chosen carefully. Following \citet{chen1979} and \citet{hannart2014} we use empirical Bayes and estimate $\delta$ by the value maximizing the marginal (or integrated) likelihood of the model, i.e. by
\begin{equation*}
	\label{optLogML}
	\widehat{\delta} = \underset{\delta}{\arg\max}\ p(Y ; \delta),
\end{equation*}
where
\begin{equation*}
	\label{logML}
	\begin{split}
	p(Y ;\delta) = \pi^{-(np)/2} \frac{\Gamma_p\left( \frac{\delta+n}{2} \right)}{\Gamma_p\left(\frac{\delta}{2}\right)} \frac{\mid \left(\delta-p-1\right) D \mid^{\frac{\delta}{2}}}{ \mid \left(\delta-p-1\right) D + S \mid^{\frac{\delta+n}{2}} }.
	\end{split}
\end{equation*}

The above optimization problem is easily solved because the marginal likelihood is concave in $\delta$. Moreover, remark that $\mid \left(\delta-p-1\right) D \mid = \left(\delta-p-1\right)^p \prod_{r=1}^{R}{d_r}$ and $\mid \left(\delta-p-1\right) D + S \mid = \left(\prod_{r=1}^{R}{d_r}\right) \left( \prod_{l=1}^{L}{\left(\delta-p-1 + e_l\right)}\right)$, where $d_r$ and $e_l$ are respectively the r$^{\text{th}}$  and l$^{\text{th}}$ largest eigenvalues of $D$ and $D^{-1}S$. Hence, evaluating the objective function for different values of $\delta$ is computationally cheap provided the eigenvalues of $D$ and $D^{-1}S$ have been pre-computed. We are referring the reader to \citet[Section 2.3.]{hannart2014} for the proof that the asymptotic properties of estimator~\eqref{expectSigma} and~\eqref{expectOmega} (Proposition~\ref{prop1}) hold when $\delta=\widehat{\delta}$.

\section{Bayes factors}
\label{BF}

\subsection{Bayes factor for conditional independence}
\label{BF:cond}
In this section we derive an analytic expression for the Bayes factor evaluating the null hypothesis of conditional independence between two variables in context of model \eqref{GC}. For ease of notation we define $F=(\delta-p-1)D$ and $T=F+S$. We wish to evaluate the null hypothesis of conditional independence, denoted $\text{H}_{0,ij}^{\text{C}}$, between two coordinates $i$ and $j$, $1\leq i<j\leq p$. We test $\text{H}_{0,ij}^{\text{C}}: \omega_{ij} = 0$ against the alternative hypothesis $\text{H}_{1,ij}^{\text{C}} : \omega_{ij} \neq 0$, where $\omega_{ij}$ is the $(i,j)^{\text{th}}$ element of $\Omega$. The Bayes factor evaluating evidence in favour of $\text{H}_{1,ij}^{\text{C}}$~is
\begin{equation}
\label{BFC1}
\text{BF}_{ij}^{\text{C}} = \frac{\int p_1(Y \mid \Sigma) p_1(\Sigma) d\Sigma}{\int p_0(Y \mid \Sigma^0) p_0(\Sigma^0) d\Sigma^0},
\end{equation}
where, by definition, $\Sigma_0$ is such that $\omega_{ij}=0$. 

\citet{giudici1995} showed that \eqref{BFC1} could be obtained in closed-form by reparameterising the Gaussian conjugate model and defining a compatible prior under the null hypothesis using the approach of \citet{dickey1971}. However, the proposed Bayes factor does not exist in high dimensional settings because it depends on elements of $S^{-1}$. This problem is here circumvented by factorising the joint likelihood of the observed data as
$$ p(Y\vert \Sigma) = p(Y_b \vert \Sigma_{bb}) p(Y_a \vert Y_b, B_{a\vert b}, \Sigma_{aa.b}) ,$$
the product of a marginal and conditional likelihood. This factorisation arises from the partition of $Y=\left[Y_a, Y_b\right]$ into two disjoint subsets indexed by $a = \{ i, j\}$ and $b=V \setminus a$. The quantity $B_{a\vert b} = \Sigma_{bb}^{-1} \Sigma_{ba}$ represents the matrix of regression coefficients obtained when regressing the variables indexed by $a$ onto the variables indexed by $b$, whereas $\Sigma_{aa.b} = \Sigma_{aa} - \Sigma_{ab} \Sigma_{bb}^{-1} \Sigma_{ba}$ denotes the residual covariance matrix. 

The factorisation of the likelihood allows conveniently to simplify \eqref{BFC1}. Using the change of variable from $(\Sigma_{aa}, \Sigma_{ab}, \Sigma_{bb})$ to $(\Sigma_{aa.b}, B_{a\mid b}, \Sigma_{bb})$ together with the fact that $\Sigma_{bb}$ is independent of $(B_{a\mid b}, \Sigma_{aa.b})$,  most nuisance parameters are integrated out and equation \eqref{BFC1} becomes
\begin{equation}
	\label{BFC2}
	\text{BF}_{ij}^{\text{C}} = \frac{\iint p_1(Y_a \mid Y_b, B_{a\mid b}, \Sigma_{aa.b}) p_1(B_{a\mid b}, \Sigma_{aa.b}) dB_{a\mid b} d\Sigma_{aa.b}}{\iint p_0(Y_a \mid Y_b, B_{a\mid b}, \Sigma_{aa.b}^0) p_0(B_{a\mid b}, \Sigma_{aa.b}^0) dB_{a\mid b} d\Sigma_{aa.b}^0}.
\end{equation}

Note that by the standard properties of the multivariate normal and Inverse-Wishart distributions \citet[Theorems 2.3.12. and 3.3.9.]{gupta2000} the densities under the alternative model are
\begin{equation}
	\label{model1alt}
	\begin{split}
		\text{vec}(Y_a) \mid Y_b, B_{a\mid b}, \Sigma_{aa.b} &\sim N_{n\times 2}\left(\text{vec}(Y_b B_{a\mid b}), \Sigma_{aa.b} \otimes I_n\right),\\
		\text{vec}(B_{a\mid b})\mid \Sigma_{aa.b} &\sim N_{(p-2)\times  2}\left(\text{vec}(F_{a\mid b}), \Sigma_{aa.b} \otimes F_{bb}^{-1}\right),\\
		\Sigma_{aa.b} &\sim IW_{2} \left(F_{aa.b}, \delta\right),
	\end{split}
\end{equation}
where $F_{a\mid b} = F_{bb}^{-1} F_{ba}$ and $F_{aa.b} = F_{aa}-F_{ab} F_{bb}^{-1} F_{ba}$. Therefore, the simplification of Bayes factor \eqref{BFC1} intuitively tells us that evaluating the conditional independence between any two coordinates within the $p$-dimensional Gaussian conjugate model \eqref{GC} is equivalent to evaluating the diagonality of the residual covariance matrix in a bivariate response regression model.

To obtain \eqref{BFC2} in closed-form we, similarly to \citet{giudici1995}, define a compatible prior for $(B_{a\mid b}, \Sigma_{aa.b})$ under the null hypothesis $\text{H}_{0,ij}^{\text{C}}$ using the conditional approach of \citet{dickey1971}. Precisely, the prior density under $\text{H}_{0,ij}^{\text{C}}$ is derived from that under $\text{H}_{1,ij}^{\text{C}}$ by conditioning on $\text{H}_{0,ij}^{\text{C}}$. The densities under the null model are therefore
\begin{equation}
	\label{model1null}
	\begin{split}
		\text{vec}(Y_a) \mid Y_b, B_{a\mid b}, \Sigma_{aa.b}^0 &\sim N_{n\times 2}\left(\text{vec}(Y_b B_{a\mid b}), \Sigma_{aa.b}^0 \otimes I_n\right),\\
		\text{vec}(B_{a\mid b})\mid \Sigma_{aa.b}^0 &\sim N_{(p-2)\times  2}\left(\text{vec}(F_{a\mid b}), \Sigma_{aa.b}^0 \otimes F_{bb}^{-1}\right),\\
		p_0(B_{a\mid b}, \Sigma_{aa.b}^0) &= p_1(B_{a\mid b}, \Sigma_{aa.b} \mid \text{H}_{0,ij}^{\text{C}}) \\
		& = \frac{p_1(B_{a\mid b} , \Sigma_{aa.b}, \text{H}_{0,ij}^{\text{C}})}{\iint p_1(B_{a\mid b}, \Sigma_{aa.b}, \text{H}_{0,ij}^{\text{C}}) dB_{a\mid b} d\Sigma_{aa.b}} ,
	\end{split}
\end{equation}
where $\Sigma_{aa.b}^0$ is such that $\omega_{ij}=0$.

We now state the main result of this section.
\medskip
\begin{lemma}
\label{lemma1}
	Assume \eqref{BFC2} holds with densities defined by \eqref{model1alt} and \eqref{model1null}. Then the Bayes factor in favour of $\text{H}_{1,ij}^{\text{C}}$ is
\begin{equation*}
		\text{BF}_{ij}^{\text{C}} = \frac{\Gamma\left( \frac{\delta+n}{2} \right) \Gamma\left( \frac{\delta+n-1}{2} \right) \Gamma^2\left(\frac{\delta+1}{2}\right)}{\Gamma\left(\frac{\delta}{2}\right) \Gamma\left(\frac{\delta-1}{2}\right) \Gamma^2\left( \frac{\delta+n+1}{2} \right)}  \frac{(1-r_{g_{ij}}^2)^{\frac{\delta}{2}}}{(1-r_{q_{ij}}^2)^{\frac{\delta+n}{2}}} \left(\frac{g_{ii} g_{jj}}{q_{ii} q_{jj}}\right)^{\frac{1}{2}},
\end{equation*}
with $F_{aa.b}=\left[ \begin{array}{cc} g_{ii} & g_{ij} \\ g_{ij} & g_{jj} \end{array} \right]$, $r_{g_{ij}}=g_{ij}(g_{ii}g_{jj})^{-1/2}$, $T_{aa.b}=\left[ \begin{array}{cc} q_{ii} & q_{ij} \\ q_{ij} & q_{jj} \end{array} \right]$ and $r_{q_{ij}}=q_{ij}(q_{ii}q_{jj})^{-1/2}$.
\end{lemma}
\smallskip
\begin{remark}
\label{remark1}
	In Lemma~\ref{lemma1}, the quantities $g_{ii}$ and $q_{ii}$ (resp. $g_{jj}$ and $q_{jj}$) can be thought of representing prior and posterior partial variances for coordinate $i$ (resp.~$j$), whereas $r_{g_{ij}}$ and $r_{q_{ij}}$ can be thought of representing prior and posterior partial correlations.
\end{remark}
\smallskip
\begin{remark}
\label{remark3}
	The Bayes factor proposed by \citet[lemma 3]{giudici1995}, in contrast to Lemma~\ref{lemma1}, defines the quantities $g_{ij}$ and $q_{ij}$ such that the matrices $F_{aa.b} = \{ g_{ij} \}$ and
$F_{aa.b} + S_{aa.b} = \{ q_{ij} \}$, with $S_{aa.b} = S_{aa} - S_{ab}S_{bb}^{-1}S_{ba}$. Note that here $S_{aa.b}$ only exists when $S_{bb}$ is invertible (i.e. when $n$ is large relatively to $p$) whereas $T_{aa.b}=T_{aa}-T_{ab} T_{bb}^{-1} T_{ba}$ defined in Lemma~\ref{lemma1} exists even when $p>n$ because $T$ is always positive definite (a consequence of Proposition~\ref{prop1}).
\end{remark}
\smallskip
\begin{remark}
\label{remark2}
	Standard matrix algebra \citep[Theorem 1.2.3.v]{gupta2000} tells us that $F_{aa.b}=\{(F^{-1})_{aa}\}^{-1}$ and $T_{aa.b}=\{(T^{-1})_{aa}\}^{-1}$. This means that the elements of the $2\times 2$ matrices $F_{aa.b}$ and $T_{aa.b}$ can respectively be obtained from the elements of $F^{-1}$ and $T^{-1}$. The computation of the Bayes factor in Lemma~\ref{lemma1} for all pairs of variables $(i,j)$, $1\leq i < j \leq p$, hence boils down to computing $F^{-1}$ and $T^{-1}$.
\end{remark}

\subsection{Bayes factor for marginal independence}
\label{BF:marg}
We now derive an analytic expression for the Bayes factor evaluating the null hypothesis of marginal independence between any two variables in context of model \eqref{GC}. We test $\text{H}_{0,ij}^{\text{M}}: \sigma_{ij} = 0$ against the alternative hypothesis $\text{H}_{1,ij}^{\text{M}} : \sigma_{ij} \neq 0$, where $\sigma_{ij}$ is the $(i,j)^{\text{th}}$ element of $\Sigma$. The Bayes factor evaluating evidence in favour of  $\text{H}_{1,ij}^{\text{M}}$ is
\begin{equation}
\label{BFM1}
\text{BF}_{ij}^{\text{M}} = \frac{\int p_1(Y \mid \Sigma) p_1(\Sigma) d\Sigma}{\int p_0(Y \mid \Sigma^0) p_0(\Sigma^0) d\Sigma^0},
\end{equation}
where now $\Sigma^0$ is such that $\sigma_{ij} = 0$.

We adopt a similar approach as in section~\ref{BF:cond} to obtain \eqref{BFM1} in closed-form. We first write the joint likelihood as
$$ p(Y\mid \Sigma) = p(Y_a \mid \Sigma_{aa}) p(Y_b \mid Y_a, B_{b\mid a}, \Sigma_{bb.a}) ,$$
and make a change of variable from $(\Sigma_{aa}, \Sigma_{ab}, \Sigma_{bb})$ to $(\Sigma_{aa}, B_{b\mid a}, \Sigma_{bb.a})$, where $B_{b\mid a} = \Sigma_{aa}^{-1} \Sigma_{ab}$ and $\Sigma_{bb.a} = \Sigma_{bb}-\Sigma_{ba} \Sigma_{aa}^{-1}\Sigma_{ab}$. Then, using the fact that $(B_{b\mid a}, \Sigma_{bb.a})$ is independent of $\Sigma_{aa}$ it is easily seen that the Bayes factor \eqref{BFM1} simplifies to
\begin{equation}
	\label{BFM2}
	\begin{split}
\text{BF}_{ij}^{\text{M}} &= \frac{\int p_1(Y_a \mid \Sigma_{aa}) p_1(\Sigma_{aa}) d\Sigma_{aa}}{\int p_0(Y_a \mid \Sigma_{aa}^0) p_0(\Sigma_{aa}^0) d\Sigma_{aa}^0}.
	\end{split}
\end{equation}
Here, the densities under the alternative model, by properties of the multivariate normal and Inverse-Wishart distributions, are
\begin{equation}
	\label{model2alt}
	\begin{split}
		\text{vec}(Y_a) \mid \Sigma_{aa} &\sim N_{n\times 2} (0, \Sigma_{aa} \otimes I_n),\\
		\Sigma_{aa} &\sim IW_{2} (F_{aa}, \delta-p+2),
	\end{split}
\end{equation}
whereas the densities under the null model are
\begin{equation}
	\label{model2null}
	\begin{split}
		\text{vec}(Y_a) \mid \Sigma_{aa}^0  & \sim N_{n\times 2} (0, \Sigma_{aa}^0 \otimes I_n),\\
		p_0(\Sigma_{aa}^0) &= p_1(\Sigma_{aa} \mid \text{H}_{0,ij}^{\text{M}}) = \frac{p_1(\Sigma_{aa}, \text{H}_{0,ij}^{\text{M}})}{\int p_1(\Sigma_{aa}, \text{H}_{0,ij}^{\text{M}}) d\Sigma_{aa}}.
	\end{split}
\end{equation}

We now state the following lemma.

\begin{lemma}
\label{lemma2}
	Assume \eqref{BFM2} holds with densities defined by \eqref{model2alt} and \eqref{model2null}. Then the Bayes factor in favour of $\text{H}_{1,ij}^{\text{M}}$ is
	\begin{equation*}
		\text{BF}_{ij}^{\text{M}} = \frac{\Gamma_2\left( \frac{\delta+n-p+2}{2} \right) \Gamma^2\left(\frac{\delta-p+3}{2}\right)}{\Gamma_2\left( \frac{\delta-p+2}{2} \right) \Gamma^2\left( \frac{\delta+n-p+3}{2} \right)} \ \frac{\left(1-r_{f_{ij}}^{2}\right)^{\frac{\delta-p+2}{2}}}{\left(1-r_{t_{ij}}^{2}\right)^{\frac{\delta+n-p+2}{2}}} \ \left( \frac{t_{ii} t_{jj}}{f_{ii} f_{jj}} \right)^{\frac{1}{2}} ,
\end{equation*}
with $F_{aa}=\left[ \begin{array}{cc} f_{ii} & f_{ij} \\ f_{ij} & f_{jj} \end{array} \right]$, $r_{f_{ij}}=f_{ij}(f_{ii}f_{jj})^{-1/2}$, $T_{aa}=\left[ \begin{array}{cc} t_{ii} & t_{ij} \\ t_{ij} & t_{jj} \end{array} \right]$ and $r_{t_{ij}}^{}=t_{ij}(t_{ii}t_{jj})^{-1/2}$.
\end{lemma}
\smallskip
\begin{remark}
\label{remark4}
	In Lemma~\ref{lemma2}, the quantities $f_{ii}$ and $t_{ii}$ (resp. $f_{jj}$ and $t_{jj}$) can be thought of representing prior and posterior marginal variances for coordinate $i$ (resp.~$j$), whereas $r_{f_{ij}}$ and $r_{t_{ij}}$ can be thought of representing prior and posterior marginal correlations.
\end{remark}
\smallskip
\begin{remark}
\label{remark5}
	The computation of the Bayes factor in Lemma~\ref{lemma2} for all pairs of variables $(i,j)$, $1\leq i < j \leq p$, boils down to computing $T$.
\end{remark}	

\subsection{Consistency}
\label{BF:cons}

In this section we consider the selection consistency of the Bayes factors defined in Lemma~\ref{lemma1} and~\ref{lemma2}. A Bayes factor is said to be consistent when $\lim_{n\rightarrow\infty} \text{BF}_{ij}=0$ if $\text{H}_{0,ij}$ is true and $\lim_{n\rightarrow\infty} \text{BF}_{ij}=\infty$ if $\text{H}_{1,ij}$ is true \citep{fernandez2001,casella2009,wangMaruyama2016}. In other words, the consistency property means that the true hypothesis will be selected when enough data are provided.

To prove the consistency of the Bayes factors, we make the following assumption.

\begin{assumption}
\label{ass1}
	The sample correlation matrix has a limit as $n\rightarrow\infty$ that is positive definite.
\end{assumption}

Assumption~\ref{ass1} also appears in \citet{maruyama2011}.
We now state the following result.

\begin{lemma}
\label{lemma3}
	Under Assumption 1 the Bayes factors $\text{BF}_{ij}^{\text{C}}$ and  $\text{BF}_{ij}^{\text{M}}$ are consistent in selection.
\end{lemma}

\section{Graph structure recovery}
\label{multest}

\subsection{Inference by multiple testing}
\label{multest:inf}

We propose to infer the marginal and conditional independence graphs by multiple testing of hypotheses using the Bayes factors introduced in the previous section.
Precisely, we propose to infer the edge set $E_U = \{ (i,j) \mid \omega_{ij}\neq 0 \}$ of the undirected graph $U=(V, E_U)$ with vertex set $V$ by evaluating $\text{H}_{0,ij}^{\text{C}}$ versus $\text{H}_{1,ij}^{\text{C}}$ for $1\leq i < j \leq p$. Similarly, we propose to infer the edge set $E_B  = \{ (i,j) \mid \sigma_{ij}\neq 0 \}$ of the bidirected graph $B=(V, E_B)$ by evaluating $\text{H}_{0,ij}^{\text{M}}$ versus $\text{H}_{1,ij}^{\text{M}}$ for $1\leq i < j \leq p$. On the whole, the approach consists in translating the pattern of rejected hypotheses into a graph \citep{drton2007}.

\subsection{Scaled Bayes factors}
\label{multest:scaled}

To infer either graph structure it is necessary to compare Bayes factors between all $p(p-1)/2$ pairs of variables. However, the Bayes factors defined in Lemma~\ref{lemma1} and~\ref{lemma2} are not scale-invariant (due to their last terms) and, hence, comparable between different pairs of variables. In light of this, we define scaled versions of the Bayes factors defined in Lemma~\ref{lemma1} and~\ref{lemma2} that can more appropriately rank edges of graphs $U$ and $B$. Corollary~\ref{cor1} and~\ref{cor2} summarize.
\begin{corollary}
\label{cor1}
	The scaled Bayes factor in favour of $\text{H}_{1,ij}^{\text{C}}$ is
	\begin{equation*}
		\text{sBF}_{ij}^{\text{C}} = \frac{\Gamma\left( \frac{\delta+n}{2} \right) \Gamma\left( \frac{\delta+n-1}{2} \right) \Gamma^2\left(\frac{\delta+1}{2}\right)}{\Gamma\left(\frac{\delta}{2}\right) \Gamma\left(\frac{\delta-1}{2}\right) \Gamma^2\left( \frac{\delta+n+1}{2} \right)}  \frac{(1-r_{g_{ij}}^2)^{\frac{\delta}{2}}}{(1-r_{q_{ij}}^2)^{\frac{\delta+n}{2}}},
	\end{equation*}
	with quantities defined as in Lemma~\ref{lemma1}.
\end{corollary}
\smallskip
\begin{corollary}
\label{cor2}
	The scaled Bayes factor in favour of $\text{H}_{1,ij}^{\text{M}}$ is
	\begin{equation*}
		\text{sBF}_{ij}^{\text{M}} = \frac{\Gamma_2\left( \frac{\delta+n-p+2}{2} \right) \Gamma^2\left(\frac{\delta-p+3}{2}\right)}{\Gamma_2\left( \frac{\delta-p+2}{2} \right) \Gamma^2\left( \frac{\delta+n-p+3}{2} \right)} \ \frac{\left(1-r_{f_{ij}}^{2}\right)^{\frac{\delta-p+2}{2}}}{\left(1-r_{t_{ij}}^{2}\right)^{\frac{\delta+n-p+2}{2}}},
	\end{equation*}
	with quantities defined as in Lemma~\ref{lemma2}.
\end{corollary}

\begin{remark}
\label{remark6}
	When the prior matrix $D$ is proportional to $I_p$, then $r_{f_{ij}}=0$ and $r_{g_{ij}}=0$, and the orderings provided by the scaled Bayes factors in Corollaries~\ref{cor1} and~\ref{cor2} for all pairs $(i,j)$ are identical to the orderings provided by the squares of the posterior marginal and partial correlations $r_{t_{ij}}$ and $r_{q_{ij}}$, respectively.
\end{remark}

\subsection{Multiplicity adjustment and error control}
\label{multest:mult}

To address the multiplicity problem we propose to use the tail or error probability associated with the null distribution of each scaled Bayes factor. The tail probability is closely related to the notion of a P-value: the Bayes factor is treated as a random variable and its distribution, which follows that of the random data, is used to make a probability statement about its observed value. Then, to recover the structure of a graph, the tail probabilities obtained from all  $p(p-1)/2$ comparisons are adjusted using standard multiplicity correction procedures to control, say, the family-wise error or false discovery rates \citep{goeman2014}.

In the following, we study the conditional null distribution of the Bayes factors statistics defined in Corollaries~\ref{cor1} and~\ref{cor2}. The conditional null distribution here refers to the distribution that would be obtained by shuffling or permuting labels of the observations \citep{jiang2017}. Under this scheme, we shall define $Pr\left(\text{sBF}_{ij}^{\text{M}} > b_1 \right)$ and $Pr\left(\text{sBF}_{ij}^{\text{C}} > b_2 \right)$ the probabilities of observing values for the two scaled Bayes factors that are respectively larger than $b_1$ and $b_2$. Next, we show that these tail probabilities can be obtained analytically without the need of a permutation algorithm.

Before, we state three results which will be used in our argumentation.

\smallskip
\begin{proposition}
\label{prop4}
	Suppose $\Phi\sim W_2(\Sigma, d)$, where
$$\Phi =
\begin{pmatrix}
	\phi_1^2 & \phi_1 \phi_2 \varphi \\
	\phi_1 \phi_2 \varphi & \phi_2^2 
\end{pmatrix}
\quad\text{and}\quad 
\Sigma =
\begin{pmatrix}
	\sigma_1^2 & \sigma_1 \sigma_2 \rho \\
	\sigma_1 \sigma_2 \rho & \sigma_2^2 
\end{pmatrix}
$$
are parametrised in terms of their correlations $-1\leq \varphi \leq 1$ and $-1\leq \rho \leq 1$. Then, $$ (\varphi^2 \mid \rho = 0) \sim Beta(1/2, (d-1)/2).$$
\end{proposition}
\smallskip
\begin{proposition}
\label{prop5}
	The following equality holds:
	\begin{equation*}
	\begin{split}
	Y_a^T Y_a - \bar{B}_{a\mid b}^T &(Y_b^T Y_b + F_{bb}) \bar{B}_{a\mid b} + F_{ab} F_{bb}^{-1} F_{ba}=\\
	&(Y_a - Y_b F_{a\vert b})^T (I_n + Y_b F_{bb}^{-1} Y_b^T)^{-1} (Y_a - Y_b F_{a\vert b}),
	\end{split}
	\end{equation*}
	where $\bar{B}_{a\mid b}=(Y_b^T Y_b + F_{bb})^{-1} (Y_b^T Y_a + F_{ba})$.
\end{proposition}
\smallskip
\begin{proposition}
\label{prop6}
	Let $\Sigma_{aa.b}$ be fixed. Then, according to model~\eqref{model1alt} we have $$(Y_a - Y_b F_{a\vert b})^T (I_n + Y_b F_{bb}^{-1} Y_b^T)^{-1} (Y_a - Y_b F_{a\vert b})\sim W_2(\Sigma_{aa.b}, n).$$
\end{proposition}
\bigskip

We observe that the only term of the Bayes factor for marginal independence (defined in Corollary~\ref{cor1}) that depends on the data is  $$r_{t_{ij}} = \frac{(f_{ii}f_{jj})^{1/2}r_{f_{ij}} + (s_{ii}s_{jj})^{1/2}r_{s_{ij}}^{}}{(f_{ii}+s_{ii})^{1/2}(f_{jj}+s_{jj})^{1/2}},$$ via the elements of $S_{aa}=Y^T_a Y_a=\{s_{ij}\}$. Here $r_{s_{ij}}=s_{ij}(s_{ii}s_{jj})^{-1/2}$. This means that we can write
$$pr\left\lbrace\text{sBF}_{ij}^{\text{M}} > b_1 \right\rbrace = pr\left\lbrace r_{s_{ij}}^2 > c_1 \right\rbrace ,$$
where $c_1$ is a quantity that depends on $\{\delta, n, f_{ii}, f_{jj}, r_{f_{ij}}, s_{ii}, s_{jj}\}$. Now, according to our model in equation~\eqref{model2alt} it is easily verified that $S_{aa}\sim W_2(\Sigma_{aa}, n)$ and, using Proposition~\ref{prop4}, we can establish that $r_{s_{ij}}^2 \mid \text{H}_{0,ij}^{\text{M}} \sim Beta(1/2, (n-1)/2)$. The tail probability of the Bayes factor can therefore be computed exactly using $Beta(1/2, (n-1)/2)$. We remark that the definition of the type I error is here conditioning on $\{\delta, n, f_{ii}, f_{jj}, r_{f_{ij}}, s_{ii}, s_{jj}\}$.

A similar argument holds for obtaining the tail probability associated with the Bayes factor for conditional independence defined in Corollary~\ref{cor2}. The only term of the Bayes factor that depends on the data is $r_{q_{ij}}=q_{ij}(q_{ii}q_{jj})^{-1/2}$, where, we recall, $q_{ij}$ is such that $T_{aa.b}=\{ q_{ij} \}$. Proposition~\ref{prop5} suggests that we can write
$T_{aa.b} = F_{aa.b} + Z$,
with $Z=(Y_a - Y_b F_{a\vert b})^T (I_n + Y_b F_{bb}^{-1} Y_b^T)^{-1} (Y_a - Y_b F_{a\vert b})$. As a result,
$$r_{q_{ij}} = \frac{(g_{ii}g_{ii})^{1/2}r_{g_{ij}} + (z_{ii}z_{jj})^{1/2}r_{z_{ij}}}{(g_{ii}+z_{ii})^{1/2}(g_{jj}+z_{jj})^{1/2}},$$ where $Z=\{z_{ij}\}$ and $r_{z_{ij}}=z_{ij}(z_{ii}z_{jj})^{-1/2}$. This means that we can write
$$Pr\left\lbrace\text{sBF}_{ij}^{\text{C}} > b_2 \right\rbrace = Pr\left\lbrace r_{z_{ij}}^2 > c_2 \right\rbrace ,$$
where $c_2$ is a quantity that depends on $\{\delta, n, g_{ii}, g_{jj}, r_{g_{ij}}, z_{ii}, z_{jj}\}$.
Propositions~\ref{prop4} and~\ref{prop6} tell us that $Z\sim W_2(\Sigma_{aa.b}, n)$ and that $r_{z_{ij}}^2 \mid \text{H}_{0,ij}^{\text{C}} \sim Beta(1/2, (n-1)/2)$. Therefore, the tail probability of the Bayes factor can also be computed exactly using $Beta(1/2, (n-1)/2)$. We remark that the definition of the type I error is conditioning on $\{\delta, n, g_{ii}, g_{jj}, r_{g_{ij}}, z_{ii}, z_{jj}\}$.

\section{Numerical experiments}
\label{num}

\subsection{Comparison to Bayesian methods}
\label{num:bmeth}

In this section we compare the performance of our approach with other Bayesian methods. For computational reasons, we consider a moderate-dimensional problem. We generate 50 datasets of size $n\in\{25,50,100\}$ from a multivariate Gaussian distribution with mean vector $0$ and $50\times 50$ inverse covariance matrix $\Phi$. The matrix $\Phi$ is a sparse matrix which we generate from a G-Wishart distribution with scale matrix equal to the identity and $b=4$ degrees of freedom (using the function bdgraph.sim of R package BDgraph). Four different graph structures are considered which we illustrate in Figure~\ref{structure}.
\begin{figure}[h]
	\begin{minipage}[c]{0.25\linewidth}
    \centering
  		\subfigure[][Band]{
  			\includegraphics[width=1.00\textwidth]{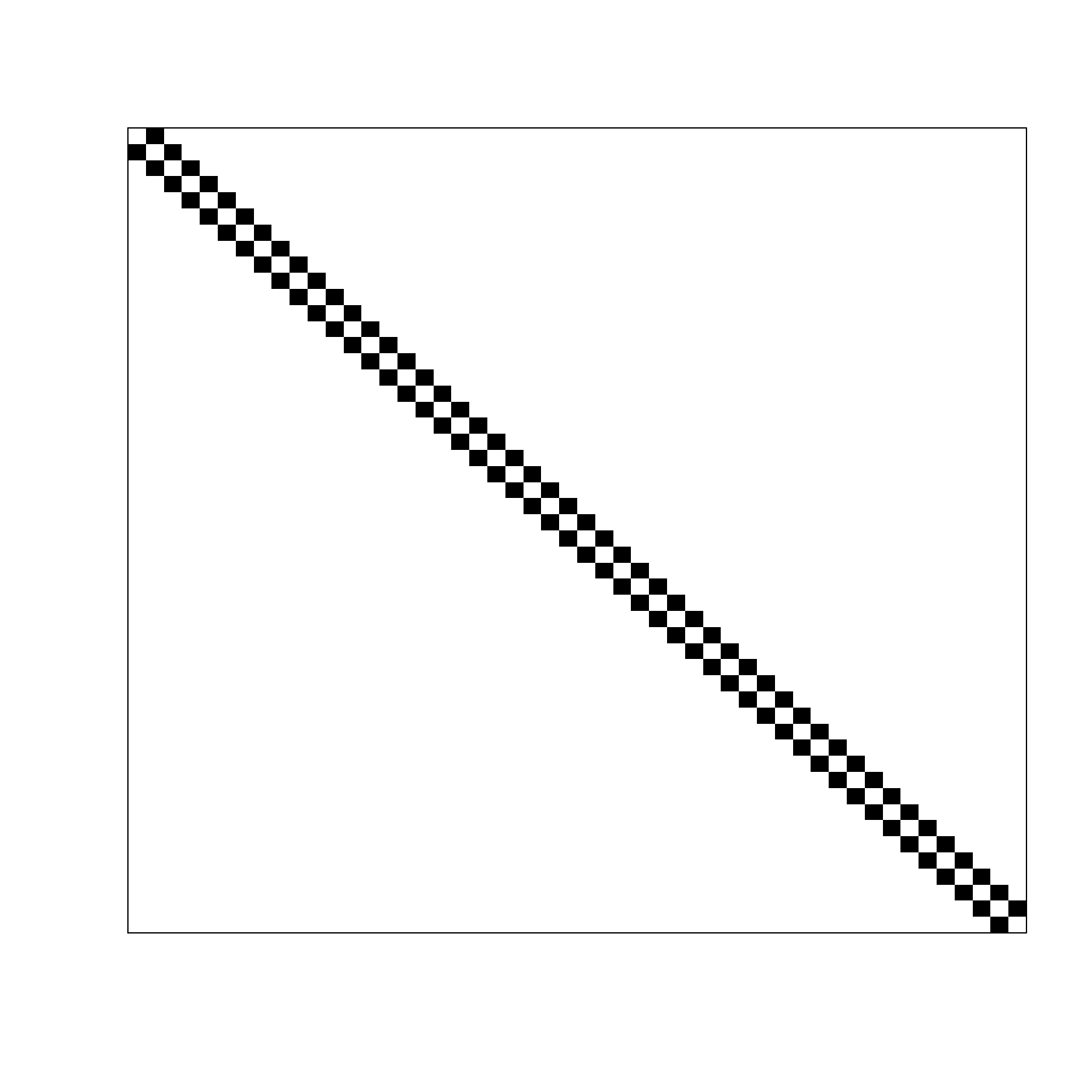}
  			\label{structure:sub1}
  		}
  \end{minipage}\hfill
	\begin{minipage}[c]{0.25\linewidth}
    \centering
  		\subfigure[][Cluster]{
  			\includegraphics[width=1.00\textwidth]{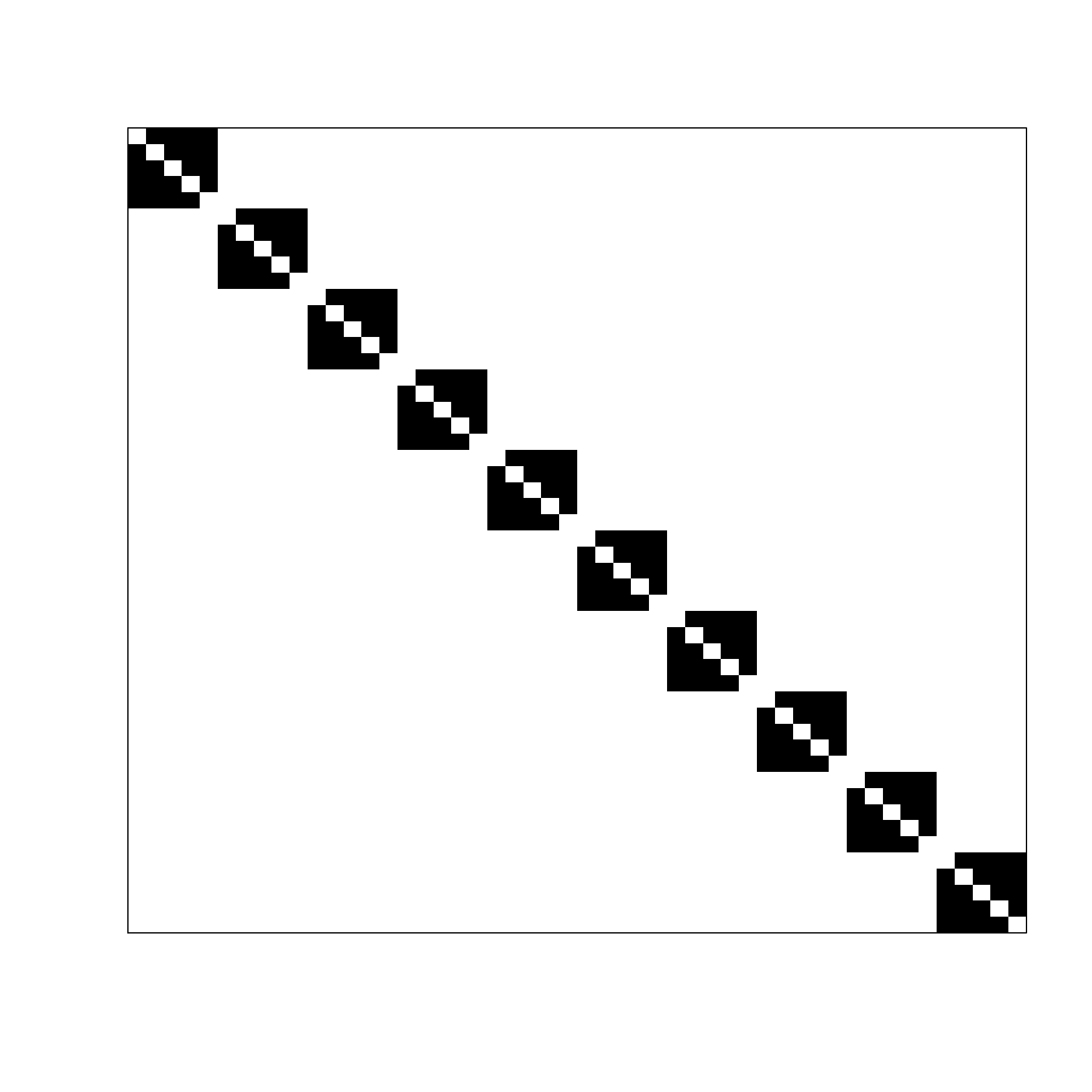}
  			\label{structure:sub2}
  		}
  \end{minipage}\hfill
  \begin{minipage}[c]{0.25\linewidth}
    \centering
  		\subfigure[][Hub]{
  			\includegraphics[width=1.00\textwidth]{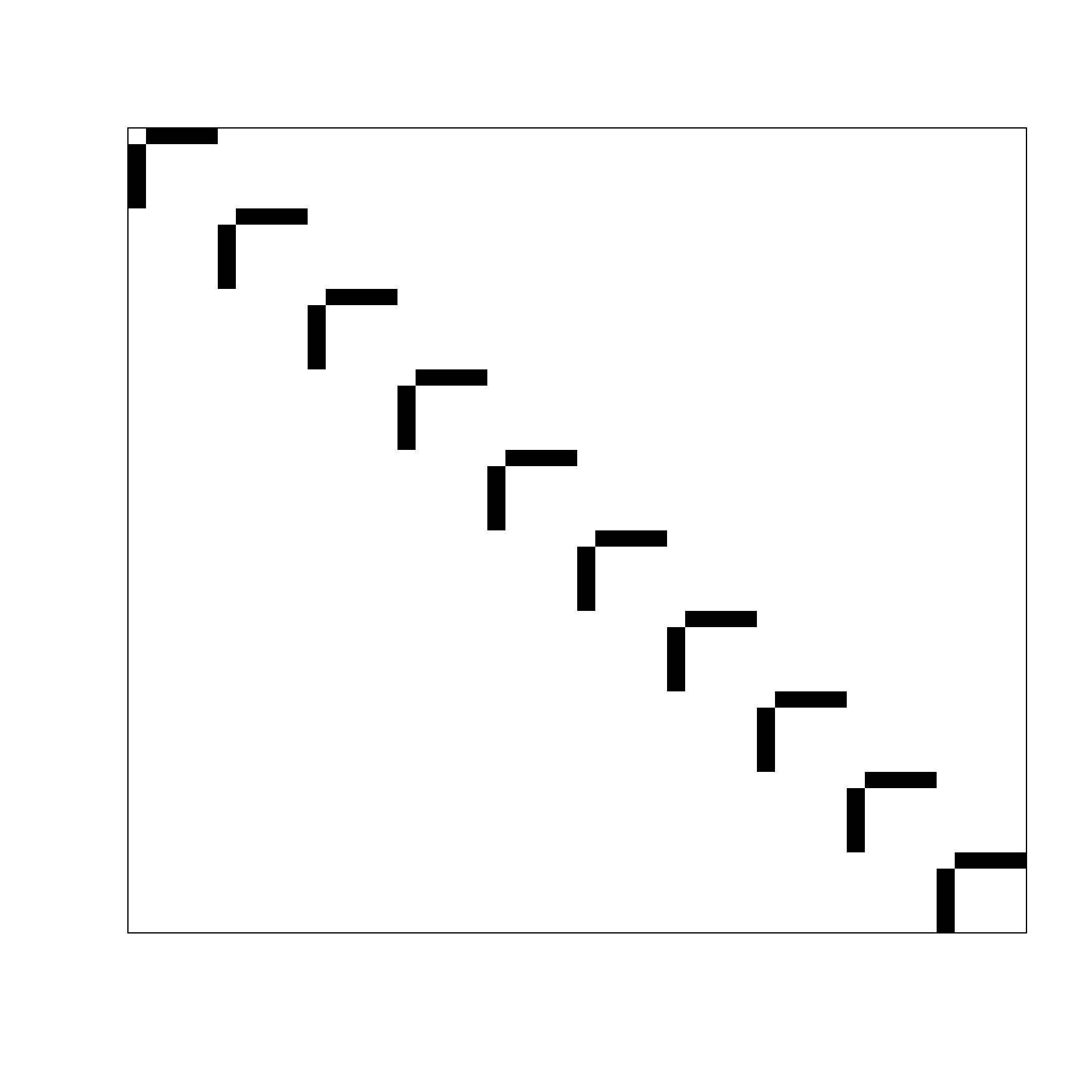}
  			\label{structure:sub3}
  		}
  \end{minipage}\hfill
  \begin{minipage}[c]{0.25\linewidth}
    \centering
  		\subfigure[][Random]{
  			\includegraphics[width=1.00\textwidth]{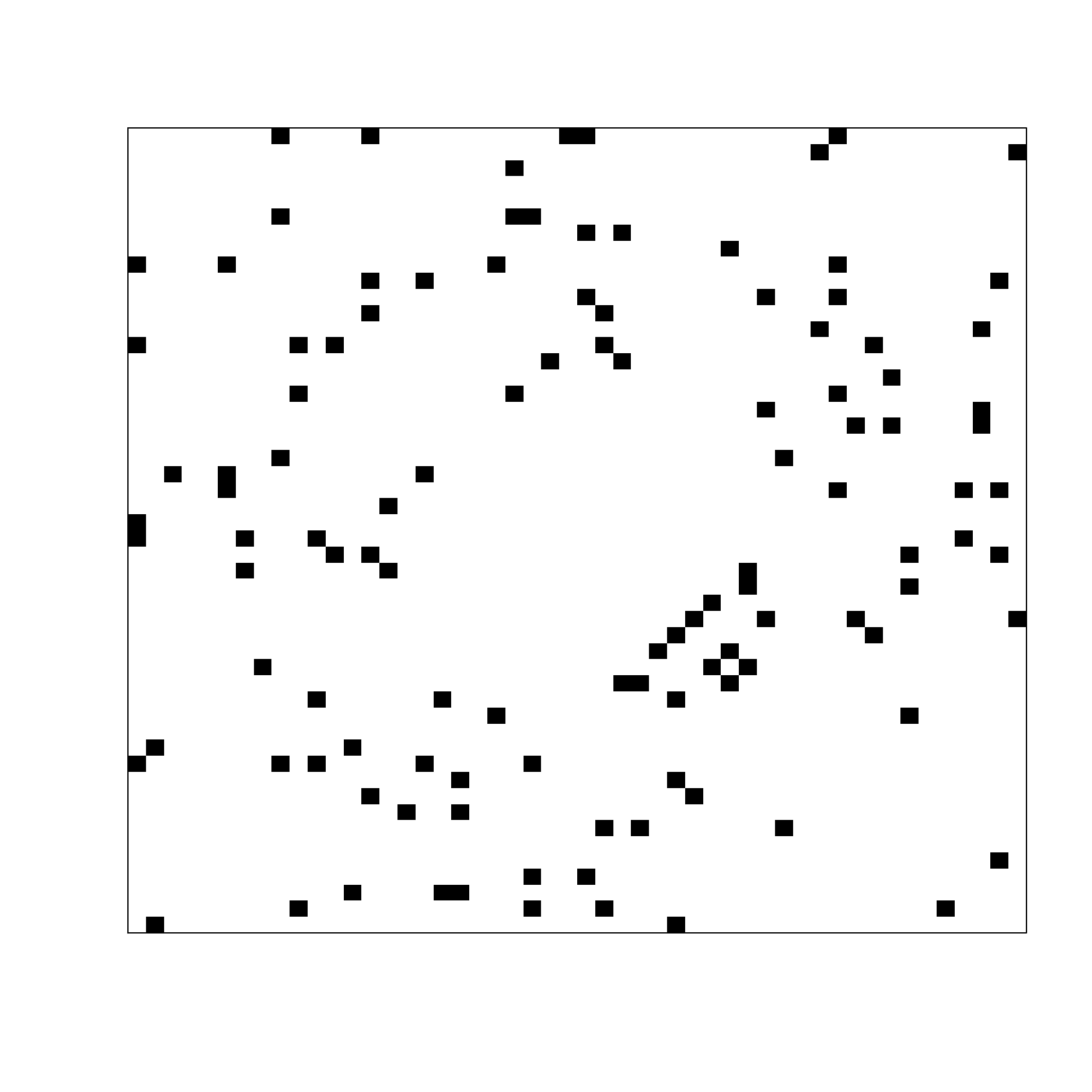}
  			\label{structure:sub4}
  		}
  \end{minipage}\hfill
		\caption{Graph structures considered in the simulation. Black and white dots represent non-zero and zero entries in $\Phi$, respectively. Only off-diagonal elements are displayed. The graph density $\eta$, that is the ratio of the number of edges and the number of possible edges, is (a) $\eta=0.080$, (b) $\eta=0.167$, (c) $\eta=0.067$ and (d) $\eta=0.083$.}
		\label{structure}
\end{figure}

We compare our method, implemented in the R package beam, to two sampling-based approaches based on the birth-death and reversible jump Markov chain Monte Carlo algorithms, developed by \citet{mohammadi2015, mohammadi2017} and implemented in the R package BDgraph, using 100,000 sweeps and a burn-in period of 50,000 updates. We also consider the method of \citet{schwaller2017}, implemented in the R package saturnin, that offers closed-form inference within the class of tree-structured graphical models. For each method we obtain the marginal posterior probabilities of edge inclusion, either via the sampling algorithm or exactly.

To evaluate the performance of the methods in recovering the different graph structures we report the area under the receiver operating characteristic  (ROC) curve which depicts the true positive rate, $TPR = TP/(TP+FN)$, as a function of the false positive rate, $FPR = (FP)/(FP+TP)$, overall possible thresholds on the marginal posterior probabilities of edge inclusion (or tail probabilities in case of our method). Here, the quantities $TP$, $FP$, $FN$ denote the number of true positives, false positives and false negatives, respectively. We also report the area under the precision-recall (PR) curve which depict the precision, $PR = TP/(TP+FP)$, as a function of the true positive rate (also referred to as recall).\\

\begin{table}[ht]\footnotesize
\centering
\begin{tabular}{cccccc}
    \hline\hline
      $n$ & Method & AUC$_\text{ROC}$ & AUC$_\text{PR}$ & AUC$_\text{ROC}$ & AUC$_\text{PR}$\\
  \hline \\[-5pt]
  & & \multicolumn{2}{c}{Band structure} & \multicolumn{2}{c}{Cluster structure}  \\[1pt]
100 & \textsc{beam} & \textbf{0.89 (0.02)} & 0.65 (0.03)  & \textbf{0.80 (0.02)} & \textbf{0.54 (0.03)} \\  
  100 & \textsc{bdmcmc} & \textbf{0.89 (0.03)} & \textbf{0.67 (0.03)}  &   0.79 (0.02) &   0.51 (0.04) \\  
  100 & \textsc{rjmcmc} & 0.88 (0.03) & 0.63 (0.05) &   0.78 (0.03) &   0.50 (0.04) \\  
  100 & \textsc{saturnin} & \textbf{0.89 (0.02)} & 0.61 (0.04) & 0.77 (0.02) & 0.53 (0.04) \\[2pt]
  50 & \textsc{beam} & \textbf{0.84 (0.03)} & \textbf{0.53 (0.04)} & \textbf{0.73 (0.02)} & \textbf{0.39 (0.04)} \\ 
  50 & \textsc{bdmcmc} &   0.82 (0.03) &   0.51 (0.06) &   0.72 (0.03) &   0.37 (0.04) \\ 
  50 & \textsc{rjmcmc} &   0.81 (0.03) &   0.47 (0.05) &   0.72 (0.02) &   0.35 (0.04) \\ 
  50 & \textsc{saturnin} & 0.82 (0.02) & 0.44 (0.04)  & 0.68 (0.02) & 0.33 (0.04) \\[2pt]
  25 & \textsc{beam} & \textbf{0.78 (0.04)} & \textbf{0.39 (0.05)} & \textbf{0.66 (0.03)} & \textbf{0.24 (0.04)} \\
  25 & \textsc{bdmcmc} &   0.75 (0.04) &   0.32 (0.05)  &   0.65 (0.03) &   0.23 (0.03) \\  
  25 & \textsc{rjmcmc} &   0.75 (0.04) &   0.27 (0.05)&   0.64 (0.03) &   0.22 (0.03) \\  
  25 & \textsc{saturnin} & 0.73 (0.03) & 0.28 (0.05)  & 0.58 (0.02) & 0.15 (0.02) \\[2pt]
  & & \multicolumn{2}{c}{Hub structure} & \multicolumn{2}{c}{Random structure}  \\[1pt]
  100 & \textsc{beam} & 0.88 (0.03) & 0.62 (0.03) & \textbf{0.87 (0.03)} & 0.65 (0.03) \\  
  100 & \textsc{bdmcmc} &   0.89 (0.02) &   \textbf{0.67 (0.04)} &   0.86 (0.03) &   \textbf{0.66 (0.03)} \\  
  100 & \textsc{rjmcmc} &   0.89 (0.02) &   0.65 (0.05) &   0.85 (0.03) &   0.65 (0.04) \\  
  100 & \textsc{saturnin} & \textbf{0.92 (0.01)} & 0.63 (0.02) & 0.86 (0.02) & 0.59 (0.02) \\[2pt] 
  50 & \textsc{beam} & 0.84 (0.03) & \textbf{0.53 (0.03)} & \textbf{0.83 (0.03)} & \textbf{0.56 (0.04)} \\  
  50 & \textsc{bdmcmc} &   0.84 (0.03) &   0.52 (0.05) &   0.81 (0.03) &   0.53 (0.05) \\  
  50 & \textsc{rjmcmc} &   0.84 (0.03) &   0.48 (0.06) &   0.80 (0.03) &   0.49 (0.06) \\  
  50 & \textsc{saturnin} & \textbf{0.86 (0.02)} & 0.48 (0.03) & \textbf{0.83 (0.02)} & 0.47 (0.03) \\[2pt]
  25 & \textsc{beam} & \textbf{0.80 (0.03)} & \textbf{0.42 (0.04)} & \textbf{0.79 (0.03)} & \textbf{0.43 (0.05)} \\ 
  25 & \textsc{bdmcmc} &   0.79 (0.04) &   0.32 (0.05) &   0.75 (0.02) &   0.33 (0.05) \\  
  25 & \textsc{rjmcmc} &   0.77 (0.04) &   0.27 (0.04) &   0.74 (0.03) &   0.30 (0.05) \\  
  25 & \textsc{saturnin} & \textbf{0.80 (0.03)} & 0.35 (0.04) & 0.77 (0.02) & 0.35 (0.04) \\
   \hline\hline
\end{tabular}
\caption{Average and standard deviation (in parenthesis) of areas under the receiver operating characteristic and precision-recall curves over the simulated datasets, as a function of the true graph structure and sample size $n$. \textsc{beam}, our method; \textsc{bdmcmc} and \textsc{rjmcmc} methods of \citet{mohammadi2015}; \textsc{saturnin} method of \citet{schwaller2017}; AUC$_\text{ROC}$, area under the receiver operating characteristic curve; AUC$_\text{PR}$ area under the precision-recall curve. Best performances are boldfaced.}
\label{simBayes}
\end{table}

Table~\ref{simBayes} summarizes simulation results. It shows that our method performs well compared to other Bayesian methods in recovering the different graph structures. For instance, our method often achieves the largest areas under the receiver operating characteristic and precision-recall curves for different graph structures and sample sizes. Moreover, a marked improvement is observed in cases where the sample size is small ($n=25$).

The results also show non-negligible differences in performance between the birth-death and reversible jump Markov chain Monte Carlo algorithms. These differences does not seem dependent of the graph structure or sample size. This suggests that the performance of the sampling-based methods can be affected by the choice of sampling algorithm.

Overall, the simulation results demonstrate that our method can recover various graphical structures at least as accurately as other Bayesian approaches at a very low computation cost (see Figure~\ref{Plottime}). Our method achieves generally a greater area under the precison-recall curve than others. The present results also confirm that obtained by \citet{schwaller2017}, namely the relative good performance of tree-structured graphical models compared to sampling-based approaches despite stronger restrictions on the class of graphical models. However, the performance of the approach can degrade in somes cases (e.g. cluster structures).

\begin{figure}[h]
	\centering
 	\includegraphics[width=0.40\textwidth]{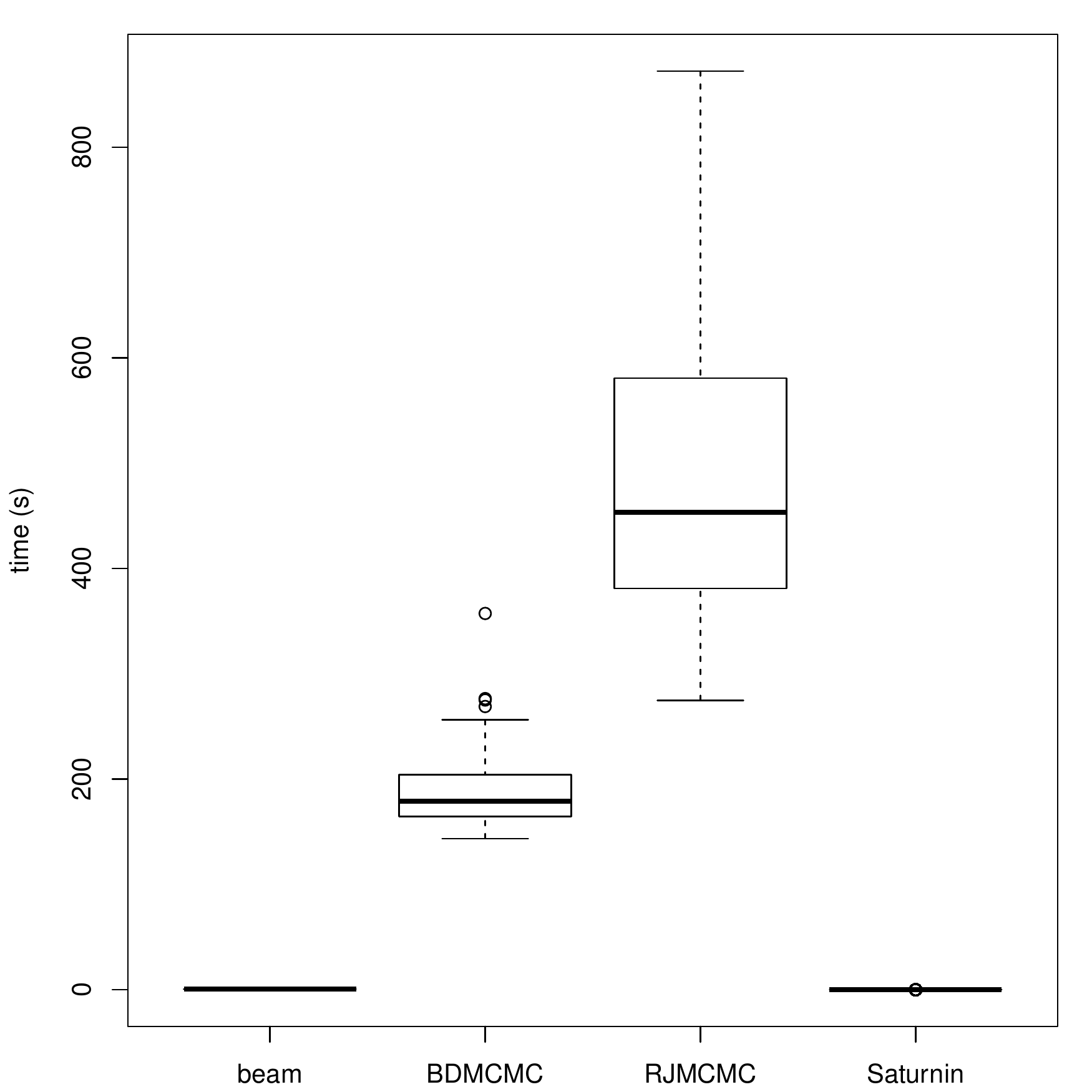}
	\caption{Running time in seconds (assessed on 3.40GHz Intel Core i7-3770 CPU) for each Bayesian method under comparison.}
	\label{Plottime}
\end{figure}

In conclusion, we remark that the marginal posterior probabilities of edge inclusion obtained from each method can in principle be linked to a Bayesian version of the false discovery rate to carry out edge selection with error control \citep{mitra2013, baladandayuthapani2014, peterson2015}. To see this, let $\pi_{ij}=Pr(\text{H}_{1,ij}^{\text{C}} \mid Y)$ be the marginal posterior probability of inclusion for edge $(i,j)$, $1<i<j<p$, for a given method. Then, its complement $Pr(\text{H}_{0,ij}^{\text{C}} \mid Y)=1-\pi_{ij}=\text{blfdr}_{ij}$ can indeed be viewed as a Bayesian version of the local false discovery rate \citep{efron2001} where the conditioning is on the data rather than a statistic. This connection serves to define the following Bayesian version of the false discovery rate \citep{newton2004}:
\begin{equation*}
	\text{BFDR}_t = E(\text{blfdr}_{ij} \mid \text{blfdr}_{ij}<t) = \frac{\sum_{i,j}{(1-\pi_{ij})I_{\pi_{ij}>t}}}{\sum_{i,j}{I_{\pi_{ij}>t}}}.
\end{equation*}
The control of the false discovery rate therefore relies heavily on the appropriate calibration of the marginal posterior probabilities of edge inclusion. The prior distribution obviously plays an important role in the quality of such calibration, however, the latter may also be affected by the sampling algorithm. Due to inherent differences between the different Bayesian approaches under comparison in this simulation study, it appears difficult to achieve a fair comparison on the control of the false discovery rate. Such comparison is therefore omitted here.

\subsection{Comparison to non-Bayesian methods}
\label{num:nbmeth}

The performance of the proposed method is compared in higher dimensional settings to non-Bayesian approaches that carry out graphical model selection via multiple testing. We generate 50 datasets of size $n=100$ from a $p$-dimensional Gaussian distribution mean vector $0$ and inverse covariance matrix $\Psi$. Throughout the simulation, we fix the sample size $n=100$ and vary of the dimensionality $p\in\{200,500,1000\}$. We consider four different sparse precision matrices corresponding to different graph structures (similar to those illustrated in Figure~\ref{structure}):
\begin{enumerate}
	\item band structure: $\Psi_p^{\text{band}}$ is a tridiagonal matrix,\\[-20pt]
	\item cluster structure: $\Psi_p^{\text{cluster}}$ is a block diagonal matrix whose diagonal blocks are sparse matrices of size $20$ where the off-diagonal entries of non-zero with probability 0.1.\\[-20pt]
	\item hub structure: $\Psi_p^{\text{hub}}$ is a block diagonal matrix whose diagonal blocks are sparse matrices of size $20$ where only the off-diagonal entries in the first row and column are non-zero,\\[-20pt]
	\item random structure: $\Psi_p^{\text{random}}$ is obtained by randomly permuting the rows and columns of $\Psi_p^{\text{band}}$.
\end{enumerate}
For all precision matrices the non-zero entries are generated independently from a uniform distribution on $[-1,1]$ and positive definiteness is ensured by adding a constant to the diagonal so that the minimum eigenvalue is equal to 0.1.

We compare our method to that of \citet{schafer2005}, implemented in the R package GeneNet, that is based on a linear shrinkage estimator of the covariance matrix \citep{ledoit2004} and a mixture model for false discovery rate estimation \citep{strimmer2008}. We also consider the asymptotic normal thresholding method of \citet{ren2015} that is implemented in the R package FastGGM \citep{wang2016}. For both methods we obtain P-values associated with the estimated partial correlations, whereas for our method we use the tail probabilities associated with the Bayes factor defined in Corollary~\ref{cor1} for all pairs of variables.

As in the previous section, we compare the performance of the methods using the areas under the receiver operating characteristic and precision-recall curves.

Table~\ref{simNonBayes} shows that the proposed method performs well in recovering large graphical structures compared to non-Bayesian methods. It achieves comparable areas under the receiver operating characteristic and precision-recall curves as other methods for different problem sizes. However, in the case of hub structures the proposed method performs better.

Besides recovering accurately the different graphical structures, Figure~\ref{Plottime2} shows that the proposed method is the fastest. When $p=1000$, the average computational time is less than a second whereas contenders are 5 to 20 times slower.\\

\begin{table}[ht]\footnotesize
\centering
\begin{tabular}{cccccc}
    \hline\hline
      $p$ & Method & AUC$_\text{ROC}$ & AUC$_\text{PR}$ & AUC$_\text{ROC}$ & AUC$_\text{PR}$ \\
  \hline \\[-5pt]
  & & \multicolumn{2}{c}{Band structure} & \multicolumn{2}{c}{Cluster structure} \\[1pt]
200 & \textsc{beam} & 0.88 (0.01) & 0.55 (0.02) & \textbf{0.91 (0.01)} & 0.58 (0.01)  \\  
  200 & \textsc{GeneNet} & \textbf{0.89 (0.01)} & \textbf{0.57 (0.02)} & \textbf{0.91 (0.01)} & 0.59 (0.01) \\  
  200 & \textsc{FastGGM} & 0.87 (0.01) & \textbf{0.57 (0.02)} & 0.89 (0.01) & \textbf{0.60 (0.02)} \\[2pt] 
  500 & \textsc{beam} & \textbf{0.91 (0.01)} & 0.58 (0.01) & \textbf{0.89 (0.01)} & 0.50 (0.01) \\  
  500 & \textsc{GeneNet} & \textbf{0.91 (0.01)} & 0.60 (0.01) & \textbf{0.89 (0.01)} & \textbf{0.52 (0.01)}\\ 
  500 & \textsc{FastGGM} & 0.90 (0.01) & \textbf{0.61 (0.01)} & 0.85 (0.01) & 0.49 (0.01) \\[2pt] 
  1000 & \textsc{beam} & \textbf{0.88 (0.01)} & 0.49 (0.01) & \textbf{0.90 (0.00)} & 0.48 (0.01) \\   
  1000 & \textsc{GeneNet} & \textbf{0.88 (0.01)} & 0.49 (0.01) & \textbf{0.90 (0.00)} & \textbf{0.49 (0.01)} \\  
  1000 & \textsc{FastGGM} &  0.87 (0.01) &  \textbf{0.51 (0.01)} &  0.87 (0.00) &  0.48 (0.01) \\[2pt] 
  & & \multicolumn{2}{c}{Hub structure} & \multicolumn{2}{c}{Random structure}  \\[1pt] 
  200 & \textsc{beam} & \textbf{0.90 (0.01)} & \textbf{0.56 (0.01)} & \textbf{0.86 (0.01)} & 0.43 (0.02) \\  
  200 & \textsc{GeneNet} & 0.85 (0.01) & 0.21 (0.03) & \textbf{0.86 (0.01)} & \textbf{0.47 (0.02)} \\ 
  200 & \textsc{FastGGM} & 0.87 (0.01) & 0.46 (0.02) & 0.85 (0.01) & \textbf{0.47 (0.02)} \\[2pt]  
  500 & \textsc{beam} & \textbf{0.92 (0.01)} & \textbf{0.54 (0.01)} & \textbf{0.82 (0.01)} & \textbf{0.35 (0.01)} \\ 
  500 & \textsc{GeneNet} & 0.90 (0.00) & 0.43 (0.01) & \textbf{0.82 (0.01)} & 0.34 (0.01) \\  
  500 & \textsc{FastGGM} & 0.88 (0.01) & 0.44 (0.01) & 0.81 (0.00) & 0.34 (0.01) \\[2pt]  
  1000 & \textsc{beam} & \textbf{0.93 (0.00)} & \textbf{0.54 (0.01)} & \textbf{0.77 (0.00)} & \textbf{0.22 (0.01)} \\ 
  1000 & \textsc{GeneNet} & 0.92 (0.00) & 0.49 (0.01) & \textbf{0.77 (0.00)} & 0.21 (0.01) \\  
  1000 & \textsc{FastGGM} &  0.89 (0.00) &  0.44 (0.01) & \textbf{0.77 (0.00)} &  \textbf{0.22 (0.01)} \\ 
   \hline\hline
\end{tabular}
\caption{Average and standard deviation (in parenthesis) areas under the receiver operating characteristic and precision-recall curves over the simulated datasets, and as a function of the true graph structure and sample size $n$. \textsc{beam}, our method; \textsc{saturnin} method of \citet{schwaller2017}; \textsc{GeneNet} method of \citet{schafer2005}; \textsc{FastGGM} method of \citet{ren2015}; AUC$_\text{ROC}$, area under the receiver operating characteristic curve; AUC$_\text{PR}$ area under the precision-recall curve.
}
\label{simNonBayes}
\end{table}

\begin{figure}[h]
	\begin{minipage}[c]{0.33\linewidth}
    \centering
  		\subfigure[][$p=200$]{
  			\includegraphics[width=1.00\textwidth]{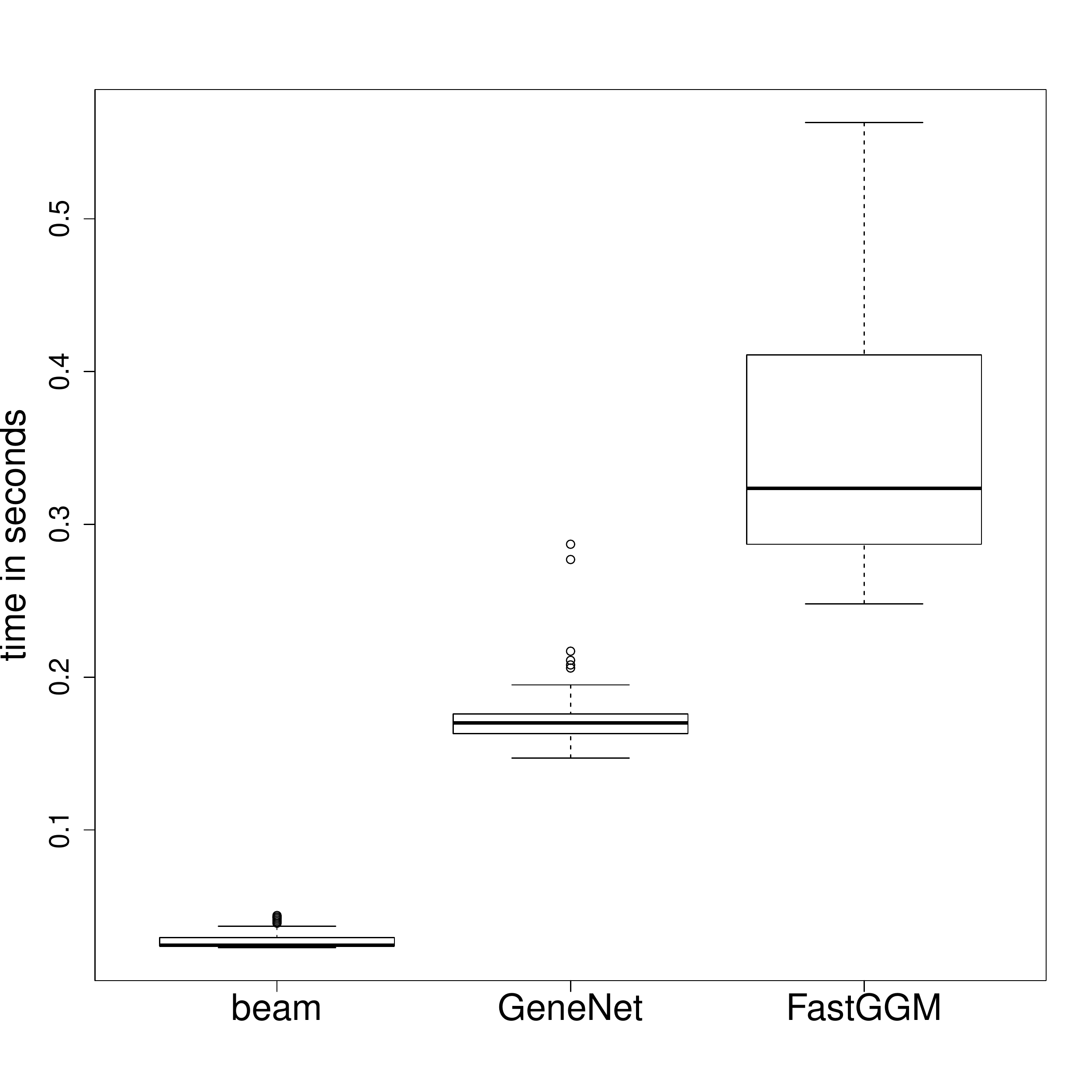}
  			\label{Plottime2:sub1}
  		}
  \end{minipage}\hfill
	\begin{minipage}[c]{0.33\linewidth}
    \centering
  		\subfigure[][$p=500$]{
  			\includegraphics[width=1.00\textwidth]{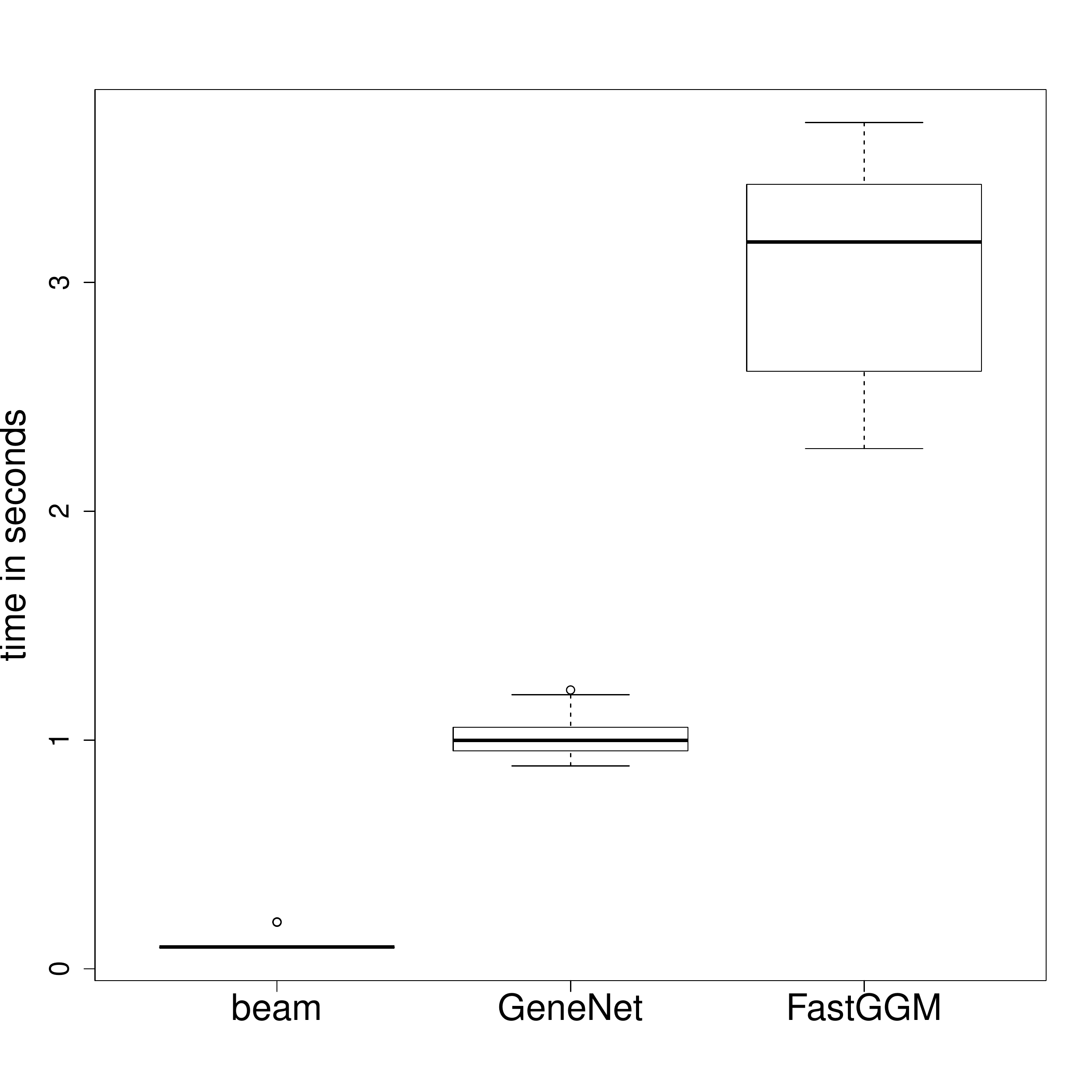}
  			\label{Plottime2:sub2}
  		}
  \end{minipage}\hfill
  \begin{minipage}[c]{0.33\linewidth}
    \centering
  		\subfigure[][$p=1000$]{
  			\includegraphics[width=1.00\textwidth]{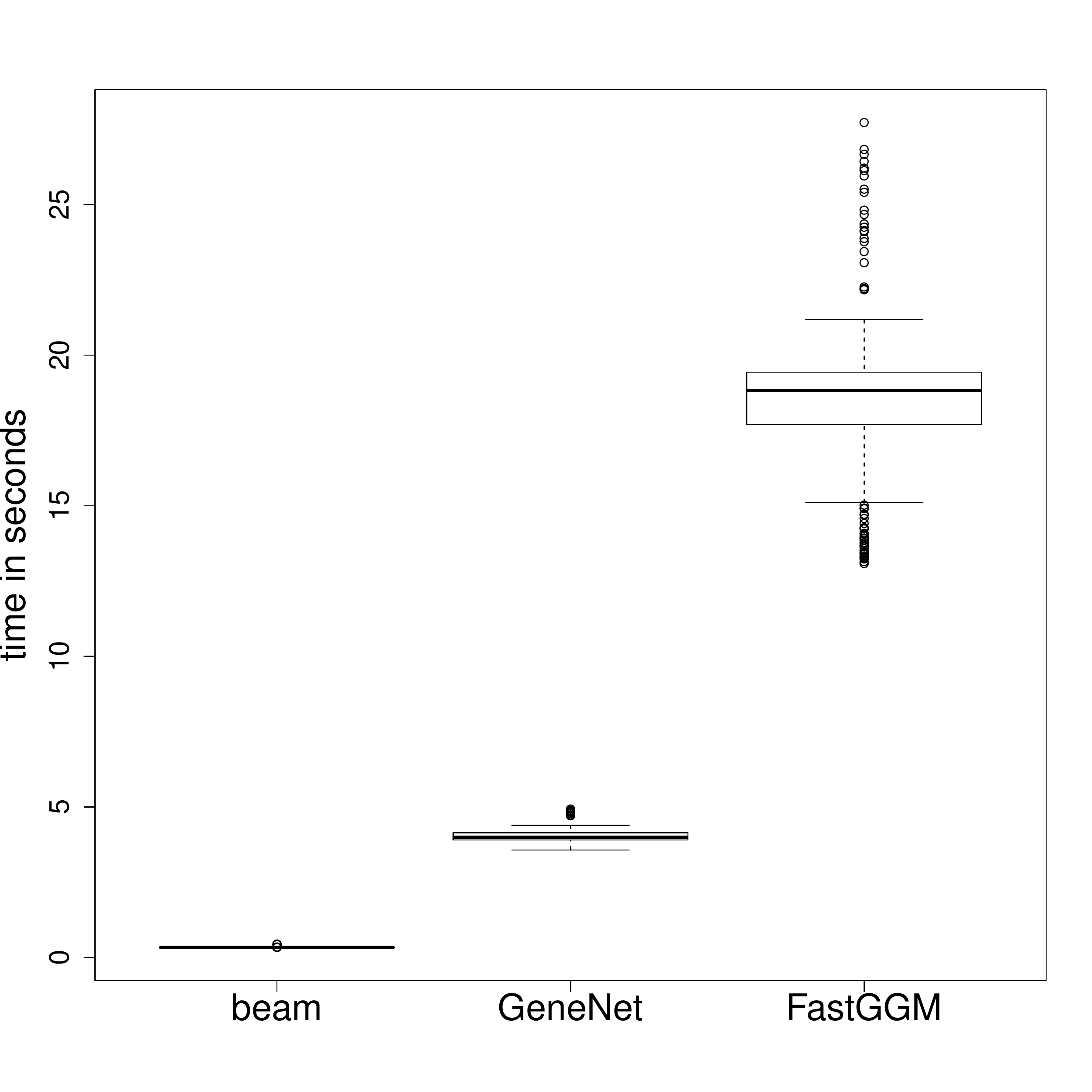}
  			\label{Plottime2:sub3}
  		}
  \end{minipage}\hfill
		\caption{Running time in seconds (assessed on 3.40GHz Intel Core i7-3770 CPU) for each method as a function of $p$.}
		\label{Plottime2}
\end{figure}

\newpage
\section{Gene network in Glioblastoma multiform}
\label{app}

We illustrate our method on a large gene expression data set on glioblastoma multiforme from The Cancer Genome Atlas. Glioblastoma multiform is an aggressive form of brain tumor in adults associated with poor prognosis. Level 3 normalized gene expression data (Agilent 244K platform) from 532 patients were obtained from The Cancer Genome Atlas Data Portal. The data comprise measurements of 17,814 genes, of which 14,827 can uniquely be identified in the PathwayCommons database. Missing expression values were imputed using the Bioconductor R package \textit{impute} (function \textit{impute.knn()} with default parameters) and the data standardized as described in section~\ref{bg:hyper}. A small subset of the data were analyzed in \citet{leday2017}. Instead, we here characterize globally the conditional independence structure between all 14,827 genes.

Figure~\ref{glio:sub1} displays the log-marginal likelihood of model~\eqref{GC} as a function of the prior parameter $\alpha$ when the prior matrix $T$ equals the identity. Using the empirical Bayes estimate of $\alpha$ we computed the Bayes factors and their associated tail probabilities for all pair of variables. These computations took 90 seconds overall on 3.40GHz Intel Core i7-3770 CPU without parallel schemes, which is remarkable for a graph with a total number of 109,912,551 possible edges.

The conditional independence graph identified by controlling the family-wise error rate at 10\% using the conservative Bonferroni procedure consists of 46,071 edges (0.042\% of the total number of edges). Edge degree varies from $0$ to $127$ with 9,675 genes having nonzero degrees. The degree distribution seems to follow an exponential distribution (see Figure~\ref{glio:sub1}), thereby indicating that a relative small number of genes have a large number of links.

\begin{figure}[h]
	\begin{minipage}[c]{0.5\linewidth}
    \centering
  		\subfigure[][]{
  			\includegraphics[width=1.00\textwidth]{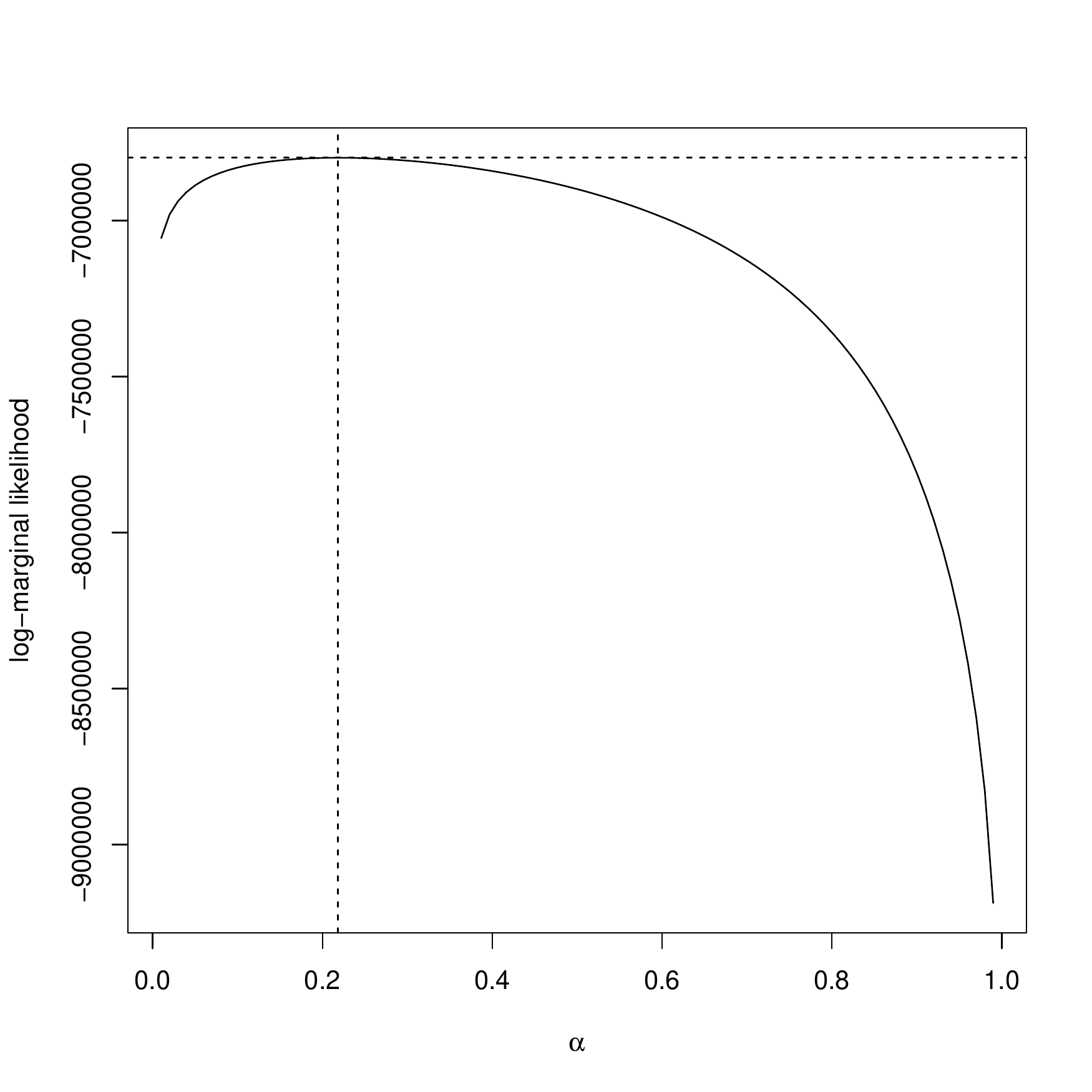}
  			\label{glio:sub1}
  		}
  \end{minipage}\hfill
	\begin{minipage}[c]{0.5\linewidth}
    \centering
  		\subfigure[][]{
  			\includegraphics[width=1.00\textwidth]{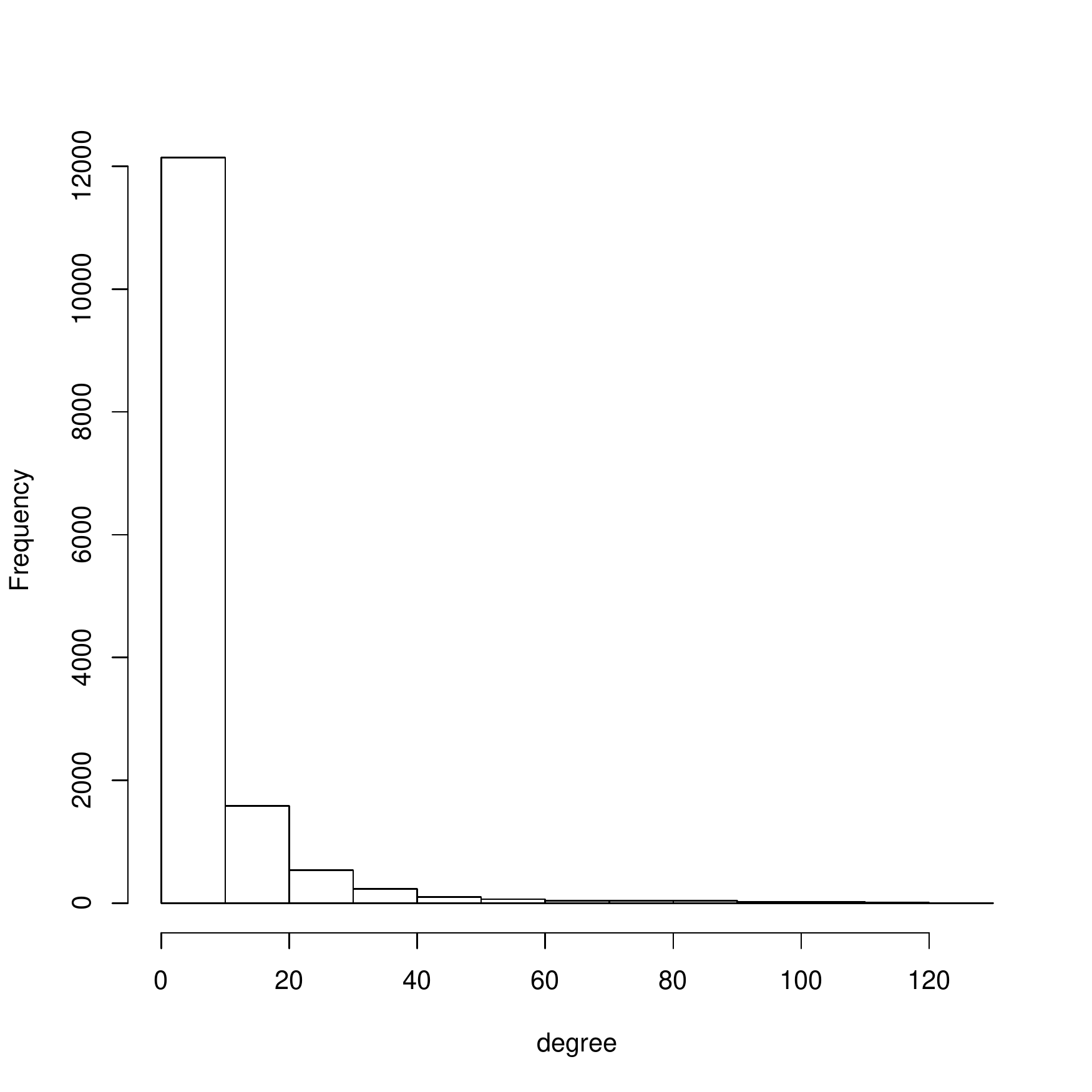}
  			\label{glio:sub2}
  		}
  \end{minipage}\hfill
		\caption{(a) Log-marginal likelihood of the Gaussian conjugate model as a function of $\alpha = (\delta-p-1)/(\delta+n-p-1)$; the vertical and horizontal dotted lines indicates the location of the optimum. (b) degree distribution of the conditional independence graph.}
		\label{glio}
\end{figure}

Because it is difficult to visualize the identified graph in its entirety, we determine clusters of densely connected edges using the algorithm of \citet{blondel2008} implemented in the R package igraph \citep{csardi2006}. The algorithm identifies a partition of edges that yield an overall modularity score equal to 0.91. The modularity score measures the quality of a division of a graph into sub-graphs. Its maximal value being 1, the identified partition presents a high modularity and suggests the presence of densely interconnected groups of nodes in the conditional independence graph. To illustrate this, we report two sub-graphs in Figure~\ref{graphs} that have been identified by the clustering algorithm and correspond to the HOXA and PCDHB gene families. The HOX gene family is known to be involved in the development of human cancers \citet{bhatlekar2014}, including Glioblastoma. The HOXA13 gene has for instance been advanced as potential diagnostic marker for Glioblastoma \citep{duan2015} and the role of HOXA9 gene in cell proliferation, apoptosis and drug resistance are under active research \citep{costa2010, gonccalves2016, bhatlekar2018}. On the other hand, the protocadherin beta (PCDHB) gene cluster, whose function is still poorly understood, have been reported to be associated with poor survival and tumour aggressiveness in Neuroblastoma \citep{banelli2015, lau2012}, another neural cancer. The particular methylation status of genes in the PCDHB family has been identified as a mechanism of transcriptional deregulation and associated with high-risk neuroblastoma biology \citep{henrich2016}.

\begin{figure}[h]
  \begin{minipage}[c]{0.5\linewidth}
    \centering
  		\subfigure[][HOXA gene cluster]{
  			\includegraphics[width=1\textwidth]{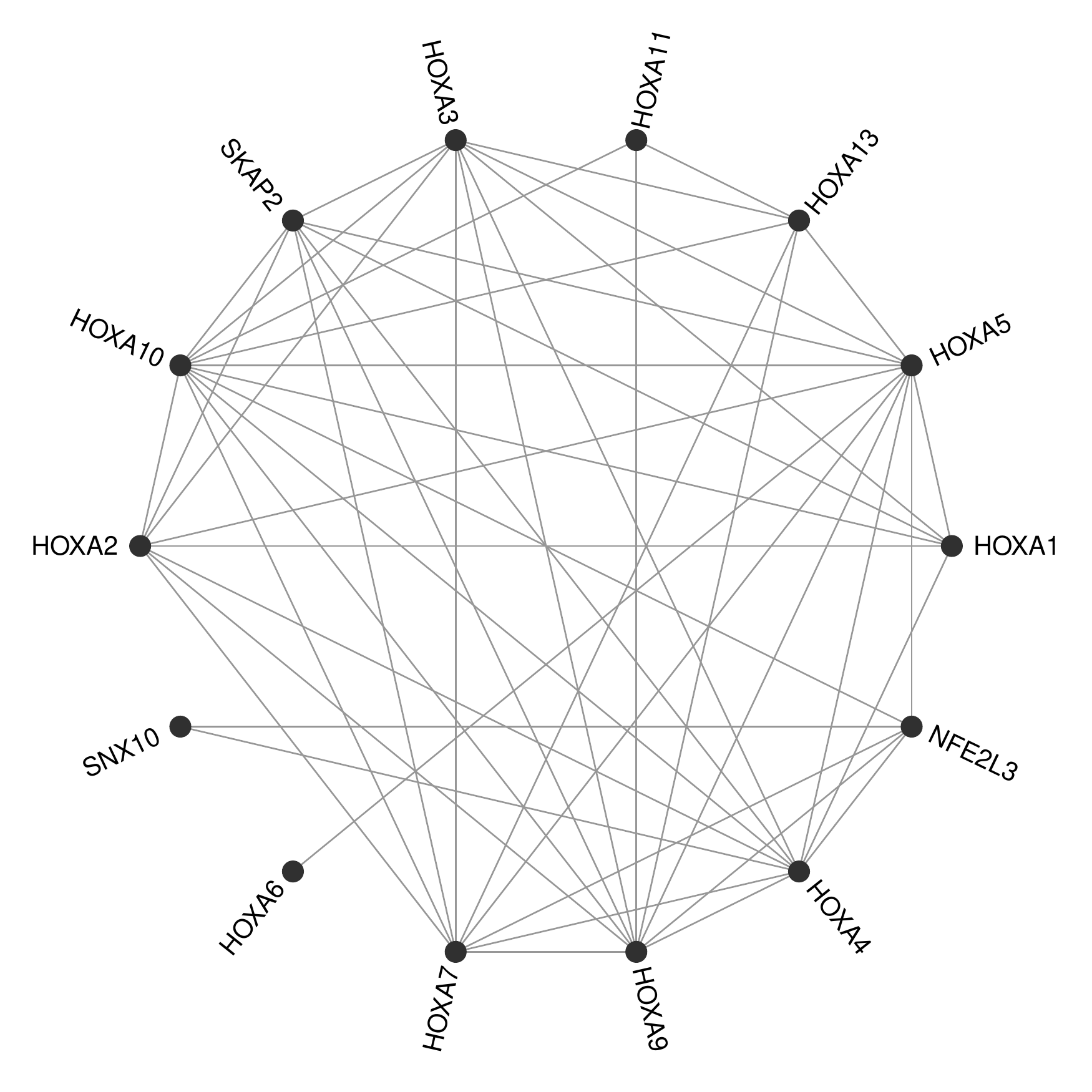}
  			\label{graphs:sub1}
  		}
  \end{minipage}\hfill
	\begin{minipage}[c]{0.5\linewidth}
    \centering
  		\subfigure[][PCDHB gene cluster]{
  			\includegraphics[width=1\textwidth]{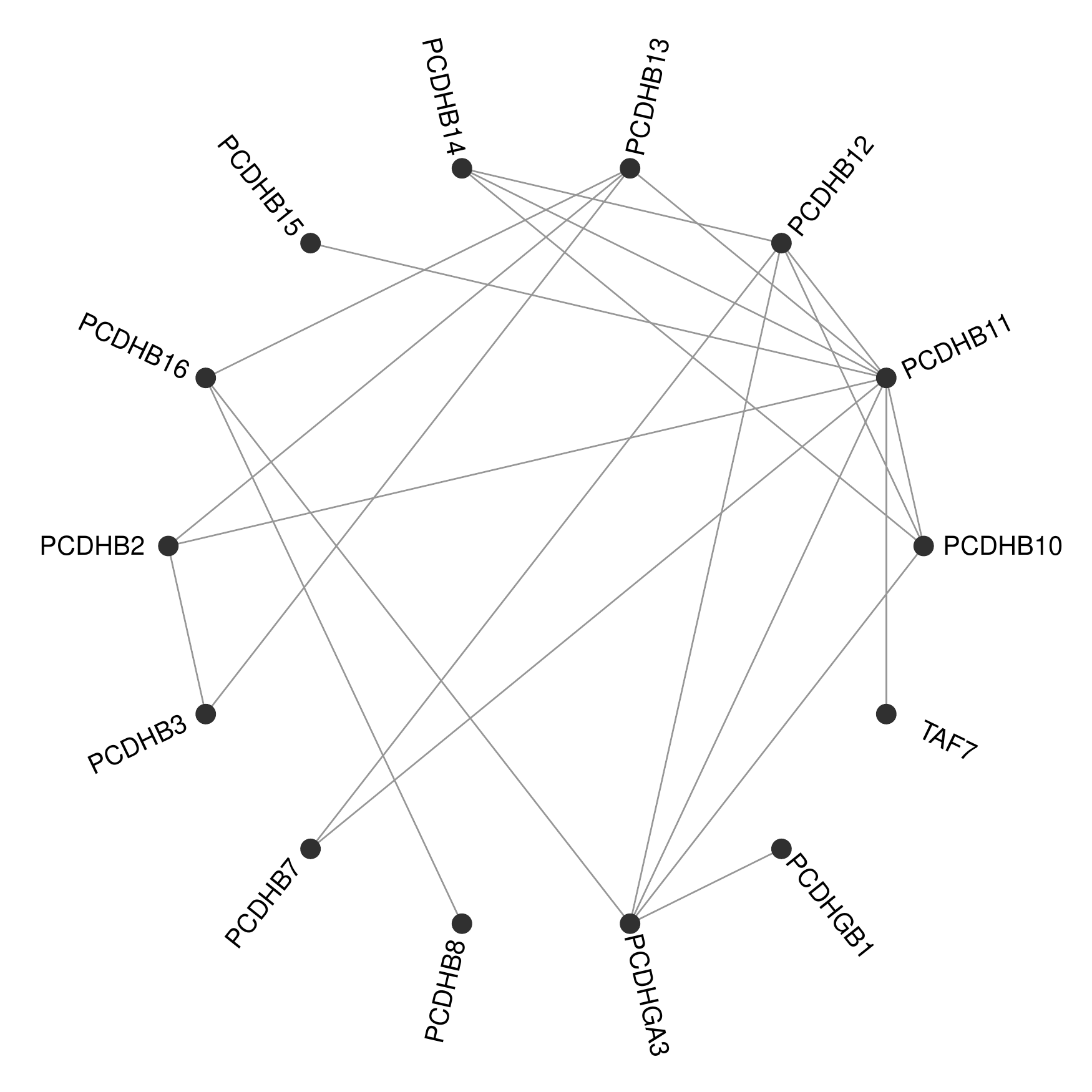}
  			\label{graphs:sub2}
  		}
  \end{minipage}\hfill
		\caption{Two examples of densely connected gene subgraphs identified by the clustering algorithm.}
		\label{graphs}
\end{figure}

\newpage
\section{Future work}
\label{disc}

We foresee several promising extensions of the proposed approach. The Bayes factors proposed in this paper can be used for differential network analysis in which the goal is to identify edges that are in common or specific to predefined groups of samples. Provided that samples between groups are independent, the Bayes factors can simply be multiplied across groups so as to obtain new Bayes factors that provide evidence towards the presence or absence of a common edge. Being symmetric, the Bayes factors can also be inverted before being multiplied so as to evaluate more complex hypotheses, e.g. edge losses or gains in a two-group comparison. Last, it would be interesting to derive the Bayes factor in a regression framework so as to compare them with that of \citet{zhou2017}.

\section*{Acknowledgements}

This research was supported by the Medical Research Council core funding number MRC\_MC\_UP\_0801/1 and grant number MR/M004421. The authors wish to thank Ilaria Speranza for helpful comments on the manuscript and improving largely the software. The first author also wishes to thank Catalina Vallejos and Leonardo Bottolo for helpful discussions.

\bibliographystyle{apalike}
\bibliography{refs}

\begin{appendices}
\section{Proofs}

This appendix contains the proofs for Lemmas~\ref{lemma1},~\ref{lemma2} and~\ref{lemma3}, as well as for Propositions~\ref{prop1},~\ref{prop2},~\ref{prop4},~\ref{prop5} and ~\ref{prop6}.\\

\begin{proof}[Proof of Proposition~\ref{prop1}]
	Let $\alpha_\delta = (\delta-p-1)/(\delta+n-p-1) \in (0,1)$ depends on $\delta$ with $n$ and $p$ fixed.
	\begin{enumerate}[(i),leftmargin=*]
		\item[(i)/(iii)] From $\widehat{\Sigma}_\delta= \alpha_\delta D + (1-\alpha_\delta)\widehat{\Sigma}_{\text{mle}}$ and the fact that $\lim_{\delta\rightarrow \infty} \alpha_\delta = 1$ and that $\lim_{\delta\rightarrow p+1} \alpha_\delta = 0$, it follows immediately that $\lim_{\delta\rightarrow \infty} \widehat{\Sigma}_\delta = D$ and $\lim_{\delta\rightarrow p+1} \widehat{\Sigma}_\delta = \widehat{\Sigma}_{\text{mle}}$ .
		\item[(ii)/(iv)] Rewrite $\widehat{\Omega}_\delta = \frac{\delta+n}{\delta+n-p-1} \widehat{\Sigma}_\delta^{-1},$
		then it clear that $\lim_{\delta\rightarrow \infty} \widehat{\Omega}_\delta = D^{-1}$ and that $\lim_{\delta\rightarrow p+1} \widehat{\Omega}_\delta = \left\lbrace(n+p+1)/n\right\rbrace\widehat{\Sigma}_{\text{mle}}^{-1}$, if $\widehat{\Sigma}_{\text{mle}}$ is positive definite.
		\item[(v)] Since by definition $x^T D x >0$ and $x^T \widehat{\Sigma}_{\text{mle}} x \geq 0$, $\forall x \in \mathbb{R}^p$, it follows that $x^T \widehat{\Sigma}_\delta x = \alpha_\delta x^T D x + (1-\alpha_\delta) x^T \widehat{\Sigma}_{\text{mle}} x >0$. Therefore, $\widehat{\Sigma}_\delta$ is positive definite for $0<\alpha_\delta \leq 1$ or equivalently $\delta > 0$. As a direct consequence, $\widehat{\Omega}_\delta = \frac{\delta+n}{\delta+n-p-1} \widehat{\Sigma}_\delta^{-1}$ is also found positive definite.
	\end{enumerate}
\end{proof}
\begin{proof}[Proof of Proposition~\ref{prop2}]
	Let $\alpha_n = (\delta-p-1)/(\delta+n-p-1) \in (0,1)$ depends on $n$ with $\delta$ and $p$ fixed.
	\begin{enumerate}[(i),leftmargin=*]
		\item[(i)] From $\widehat{\Sigma}_\delta= \alpha_n D + (1-\alpha_n)\widehat{\Sigma}_{\text{mle}}$ and the fact that $\lim_{n\rightarrow \infty} \alpha_n = 0$, it follows immediately that $\lim_{n\rightarrow \infty} \widehat{\Sigma}_n = \widehat{\Sigma}_{\text{mle}}$.
		\item[(ii)] Rewriting $\widehat{\Omega}_n = \frac{\delta+n}{\delta+n-p-1} \widehat{\Sigma}_\delta^{-1}$ it is clear that $\lim_{n\rightarrow \infty} \widehat{\Omega}_n =  \widehat{\Sigma}_{\text{mle}}^{-1}$.
	\end{enumerate}
\end{proof}
\begin{proof}[Proof of Lemma~\ref{lemma1}]
    The numerator of the Bayes factor is
\begin{equation}
	\label{num2}
	\begin{split}
		  \iint p_1(Y_a \mid Y_b, B_{a\mid b}, \Sigma_{aa.b}) p_1(B_{a\mid b}, \Sigma_{aa.b}) dB_{a\mid b} d\Sigma_{aa.b} =
		\frac{\vert F_{bb} \vert \vert F_{aa.b} \vert^{\frac{\delta}{2}} \Gamma_2\left( \frac{\delta+n}{2} \right)}{\pi^{n} \Gamma_2\left(\frac{\delta}{2}\right) \vert T_{bb} \vert \vert T_{aa.b} \vert^{\frac{\delta+n}{2}} }
	\end{split}.
\end{equation}
Under $\text{H}_{0,ij}^{\text{C}}$, the model likelihood is
\begin{equation*}
	\begin{split}
		p_0(Y_a \vert Y_b , B_{a\vert b}, \Sigma_{aa.b}^0) &= (2\pi)^{-n} (\omega_{ii}\omega_{jj})^{\frac{n}{2}} \exp\left\lbrace -\frac{1}{2} \omega_{ii} (Y_a^{(i)} - Y_b B_{a\vert b}^{(i)} )^T (Y_a^{(i)} - Y_b B_{a\vert b}^{(i)}) \right\rbrace\\
		& \quad\times \exp\left\lbrace -\frac{1}{2} \omega_{jj} (Y_a^{(j)} - Y_b B_{a\vert b}^{(j)} )^T (Y_a^{(j)} - Y_b B_{a\vert b}^{(j)}) \right\rbrace
	\end{split},
\end{equation*}
and the probability density of $(\Sigma_{aa.b}^0, B_{a\vert b})$ is
\begin{equation*}
	\begin{split}
		p_0(\Sigma_{aa.b}^0, B_{a\vert b}) & = \frac{\vert F_{bb} \vert \left(g_{ii} g_{jj}\right)^{\frac{\delta+1}{2}}}{\pi^{(p-2)} 2^{\delta+p-1} \Gamma^2\left(\frac{\delta+1}{2}\right)} (\omega_{ii} \omega_{jj})^{\frac{\delta+p+1}{2}} \exp\left\lbrace -\frac{\omega_{ii} g_{ii} + \omega_{jj} g_{jj}}{2} \right\rbrace \\
		& \times \exp\left\lbrace -\frac{1}{2} \left(B_{a\vert b}^{(i)} - F_{a\vert b}^{(i)} \right)^T \omega_{ii} F_{bb} \left(B_{a\vert b}^{(i)} - F_{a\vert b}^{(i)} \right) \right\rbrace \\
		& \times \exp\left\lbrace -\frac{1}{2} \left(B_{a\vert b}^{(j)} - F_{a\vert b}^{(j)} \right)^T \omega_{jj} F_{bb} \left(B_{a\vert b}^{(j)} - F_{a\vert b}^{(j)} \right) \right\rbrace
	\end{split},
\end{equation*}
where $B_{a\vert b}^{(l)}$ and $Y_a^{(l)}$ represent the $l^{\text{th}}$ column of $B_{a\vert b}$ and $Y_a$. Therefore,
\begin{equation}
	\label{denum2}
	\begin{split}
		\iint p_0(Y_a \vert Y_b , B_{a\vert b}, \Sigma_{aa.b}^0)  & p_0(\Sigma_{aa.b}^0, B_{a\vert b}) dB_{a\vert b} d\Sigma_{aa.b}^0 = 
		\frac{\vert F_{bb} \vert \Gamma^2\left( \frac{\delta+n+1}{2}\right) \left(g_{ii}g_{jj}\right)^{\frac{\delta+1}{2}} }{\pi^{n} \vert T_{bb} \vert \Gamma^2\left(\frac{\delta+1}{2}\right)  (q_{ii} q_{jj})^{\frac{\delta+n+1}{2}}}
	\end{split}.
\end{equation}
Combining \eqref{num2} and \eqref{denum2} we obtain the Bayes factor in Lemma~\ref{lemma1}.
\end{proof}
\null\bigskip

\begin{proof}[Proof of Lemma~\ref{lemma2}]
    The numerator of the Bayes factor is
\begin{equation}
		\label{num10}
		\int p_1(Y_a \vert \Sigma_{aa}) p_1(\Sigma_{aa}) d\Sigma_{aa} =
		\frac{\vert F_{aa} \vert^{\frac{\delta-p+2}{2}} \Gamma_2\left( \frac{\delta+n-p+2}{2} \right)}{\pi^{n} \Gamma_2\left( \frac{\delta-p+2}{2} \right) \vert T_{aa} \vert^{\frac{\delta+n-p+2}{2}}} .
\end{equation}	
Under $\text{H}_{0,ij}^{\text{M}}$, the model likelihood is
$$
		p_0(Y_a \vert \Sigma_{aa}^0) = (2\pi)^{-n} (\sigma_{ii} \sigma_{jj})^{-\frac{n}{2}} \exp\left\lbrace -\frac{1}{2} \left( \sigma_{ii}^{-1} s_{ii} + \sigma_{jj}^{-1} s_{jj}\right) \right\rbrace,
$$
and the probability density of $\Sigma_{aa}^0$ is
$$
		p_0(\Sigma_{aa}^0) = \frac{(f_{ii} f_{jj})^\frac{\delta-p+3}{2}}{2^{\delta+1} \Gamma^2\left(\frac{\delta-p+3}{2}\right)} (\sigma_{ii} \sigma_{jj})^{-\frac{\delta-p+5}{2}} \exp\left\lbrace -\frac{1}{2} \left( f_{ii} \sigma_{ii}^{-1} + f_{jj} \sigma_{jj}^{-1} \right) \right\rbrace,
$$
As result,
\begin{equation}
	\label{denum10}
	\begin{split}
	\int p_0(Y_a &\vert \Sigma_{aa}^0) p_0(\Sigma_{aa}^0) d\Sigma_{aa}^0 = 
	\frac{ \Gamma^2\left( \frac{\delta+n-p+3}{2} \right) (f_{ii} f_{jj})^\frac{\delta-p+3}{2}}{\pi^{n} \Gamma^2\left(\frac{\delta-p+3}{2}\right) (t_{ii} t_{jj})^\frac{\delta+n-p+3}{2}}
	\end{split}.
\end{equation}
Combining \eqref{num10} and \eqref{denum10} we obtain the Bayes factor in Lemma~\ref{lemma2}.
\end{proof}
\begin{proof}[Proof of Lemma~\ref{lemma3}]
Using Stirling's formula, the gamma function can asymptotically be approximated by
$$ \Gamma(\gamma_1 x + \gamma_2) \approx \sqrt{2\pi}\exp\left\lbrace -\gamma_1 x\right\rbrace (\gamma_1 x)^{\gamma_1 x+\gamma_2 -1/2},$$
for large values of $x$ \citep{wangMaruyama2016}. This means that for large values of $n$,
\begin{equation*}
	\begin{split}
	\Gamma_2\left( \frac{\delta+n-p+2}{2} \right) &\approx 2 \pi^{3/2} \exp\left\lbrace-n\right\rbrace \left( \frac{n}{2} \right)^{\frac{2n+2\delta-2p+1}{2}} 
	\end{split}
\end{equation*}
and
\begin{equation*}
	\begin{split}
	\Gamma^2\left( \frac{\delta+n-p+3}{2} \right) &\approx 2\pi \exp\left\lbrace-n\right\rbrace \left( \frac{n}{2} \right)^{n+\delta-p+2} 
	\end{split}.
\end{equation*}
The Bayes factor in Lemma~\ref{lemma2} is therefore asymptotically equivalent to
\begin{equation*}
	\begin{split}
	\text{BF}_{ij}^{\text{M}} &\approx \frac{\Gamma_2\left( \frac{\delta+n-p+2}{2} \right)}{\Gamma^2\left( \frac{\delta+n-p+3}{2} \right)} \ \left(1-r_{s_{ij}}^{2}\right)^{-\frac{\delta+n-p+2}{2}} \approx \frac{(1-r_{s_{ij}}^{2})^{-\frac{\delta+n-p+2}{2}}}{n^{3/2}}
	\end{split}
\end{equation*}
Now when $\text{H}^{\text{M}}_{0,ij}$ is true, $\lim_{n\rightarrow\infty} r^2_{s_{ij}}=0$ because the sample correlation is asymptotically unbiased. Hence, $\lim_{n\rightarrow\infty} (1-r_{s_{ij}}^{2})^{-\frac{\delta+n-p+2}{2}}=1$. Since $\lim_{n\rightarrow\infty} n^{3/2}=\infty$, we can conclude that $\lim_{n\rightarrow\infty} \text{BF}_{ij}^{\text{M}}=0$, which proves the consistency under $\text{H}^{\text{M}}_{0,ij}$. On the other hand, when $\text{H}^{\text{M}}_{1,ij}$ is true, $\lim_{n\rightarrow\infty} (1-r_{s_{ij}}^{2})^{-1}=c_1$, $c_1>0$. The application of L'H\^{o}pital's rule twice, by deriving the numerator and denominator twice with respect to $n$, allows us to conclude that $\lim_{n\rightarrow\infty} \text{BF}_{ij}^{\text{M}}=\infty$. This completes our proof for the consistency of $\text{BF}_{ij}^{\text{M}}$.

Using similar arguments we prove the consistency of the Bayes factor in Lemma~\ref{lemma1}. The latter is asymptotically equivalent to
\begin{equation*}
	\begin{split}
	\text{BF}_{ij}^{\text{C}} &\approx \frac{\Gamma\left( \frac{\delta+n}{2} \right) \Gamma\left( \frac{\delta+n-1}{2} \right)}{\Gamma^2\left( \frac{\delta+n+1}{2} \right)}  (1-r_{q_{ij}}^2)^{-\frac{\delta+n}{2}} \approx \frac{(1-r_{p_{ij}}^2)^{-\frac{\delta+n}{2}}}{n^{3/2}},
	\end{split}
\end{equation*}
where $r_{p_{ij}}$ denotes the sample partial correlation between variables $i$ and $j$. On one hand, when $\text{H}^{\text{C}}_{0,ij}$ is true $\lim_{n\rightarrow\infty} r^2_{p_{ij}}=0$ because the sample partial correlation is asymptotically unbiased. This implies that $\lim_{n\rightarrow\infty} (1-r_{p_{ij}}^{2})^{-\frac{\delta+n-p+2}{2}}=1$. And because $\lim_{n\rightarrow\infty} n^{3/2}=\infty$ we can conclude that $\lim_{n\rightarrow\infty} \text{BF}_{ij}^{\text{C}}=0$. On the other hand, when $\text{H}^{\text{C}}_{1,ij}$ is true, $\lim_{n\rightarrow\infty} (1-r_{p_{ij}}^{2})^{-1}=c_2$, $c_2>0$. Therefore, by applying L'H\^{o}pital's rule twice, as above, it is found that $\lim_{n\rightarrow\infty} \text{BF}_{ij}^{\text{C}}=\infty$, which completes our proof for the consistency of $\text{BF}_{ij}^{\text{C}}$.
\end{proof}
\null \medskip
\begin{proof}[Proof of Proposition~\ref{prop4}]
It follows directly from the probability density of $\Phi$ that $p(\varphi \mid \rho = 0) \propto (1-\varphi^2)^{\frac{d-3}{2}}$, which implies that $p(\varphi^2 \mid \rho = 0)\propto \varphi^{-1}(1-\varphi^2)^{\frac{d-3}{2}}$.
\end{proof}
\null \medskip
\begin{proof}[Proof of Proposition~\ref{prop5}]
	After some algebra we have
	\begin{equation}
	\label{equality0}
	\begin{split}
	Y_a^T Y_a - \bar{B}_{a\mid b}^T &(Y_b^T Y_b + F_{bb}) \bar{B}_{a\mid b} + F_{ab} F_{bb}^{-1} F_{ba}= \\
	& Y_a^T \left\lbrace I_n - Y_b \left(Y_b^T Y_b + F_{bb} \right)^{-1} Y_b^T \right\rbrace Y_a \\
	& - 2 Y_a^T Y_b \left(Y_b^T Y_b + F_{bb}\right)^{-1} F_{bb} F_{a\vert b} \\
	& + F_{a\vert b}^T \left\lbrace F_{bb}-F_{bb}\left(Y_b^T Y_b + F_{bb}\right)^{-1}F_{bb}\right\rbrace F_{a\vert b},
	\end{split}
	\end{equation}
where, we recall, $\bar{B}_{a\mid b}=(Y_b^T Y_b + F_{bb})^{-1} (Y_b^T Y_a + F_{ba})$.
Now,
	\begin{equation}
	\label{equality1}
	 Y_b \left(Y_b^T Y_b + F_{bb}\right)^{-1} F_{bb} = \left\lbrace I_n - Y_b \left(Y_b^T Y_b + F_{bb} \right)^{-1} Y_b^T \right\rbrace Y_b ,
	\end{equation}
and, using twice the Sherman-Morrison-Woodbury matrix identity \citep[Theorem 1.2.3.iv]{gupta2000}, it is found that
	\begin{equation}
	\label{equality2}
	\begin{split}
	F_{bb}-F_{bb}\left(Y_b^T Y_b + F_{bb}\right)^{-1}F_{bb} &=  \left\lbrace F_{bb}^{-1} + (Y_b^T Y_b)^{-1} \right\rbrace^{-1} \\
	&=Y_b^T \left\lbrace I_n - Y_b \left(Y_b^T Y_b + F_{bb} \right)^{-1} Y_b^T \right\rbrace Y_b .
	\end{split}
	\end{equation}
Thus, by plugin in~\eqref{equality1} and~\eqref{equality2} in~\eqref{equality0}, we obtain
	\begin{equation*}
	\begin{split}
	 Y_a^T Y_a - \bar{B}_{a\mid b}^T & (Y_b^T Y_b + F_{bb}) \bar{B}_{a\mid b} + F_{ab} F_{bb}^{-1} F_{ba} = \\
	& \left(Y_a - Y_b F_{a\vert b}\right)^T \left\lbrace I_n - Y_b \left(Y_b^T Y_b + F_{bb}\right)^{-1} Y_b^T\right\rbrace^{-1} \left(Y_a - Y_b F_{a\vert b}\right),
	\end{split}
	\end{equation*}
which is further reduced to
	\begin{equation*}
	\begin{split}
	Y_a^T Y_a - \bar{B}_{a\mid b}^T & (Y_b^T Y_b + F_{bb}) \bar{B}_{a\mid b} + F_{ab} F_{bb}^{-1} F_{ba} = \\
	& \left(Y_a - Y_b F_{a\vert b}\right)^T \left(I_n + Y_b F_{bb}^{-1} Y_b^T\right)^{-1} \left(Y_a - Y_b F_{a\vert b}\right),
	\end{split}
	\end{equation*}
using, again, the Sherman-Morrison-Woodbury matrix identity.
\end{proof}
\null \medskip
\begin{proof}[Proof of Proposition~\ref{prop6}]
	Consider the case where $\Sigma_{aa.b}$ is fixed, then the joint density of model~\eqref{model1alt} is
\begin{equation*}
\begin{split}
	p(Y_a, B_{a\mid b}\mid Y_b, \Sigma_{aa.b}) &\propto \exp\left\lbrace -\frac{1}{2} \text{tr}\left[ \Sigma_{aa.b}^{-1} \left( \left( Y_a - Y_b B_{a\mid b} \right)^T \left( Y_a - Y_b B_{a\mid b} \right) \right. \right. \right. \\
	& \qquad\qquad\qquad\left.\left.\left. + \left( B_{a\mid b} - F_{a\mid b} \right)^T \left( B_{a\mid b} - F_{a\mid b} \right)\right) \right] \right\rbrace \\
	&= \exp\left\lbrace -\frac{1}{2} \text{tr}\left[ \Sigma_{aa.b}^{-1} \left(\left( B_{a\mid b} - \bar{B}_{a\mid b} \right)^T \left( B_{a\mid b} - \bar{B}_{a\mid b} \right) \right. \right. \right. \\
	& \qquad\qquad\qquad\left.\left.\left. + Y_a^T Y_a - \bar{B}_{a\mid b}^T \left(Y_b^T Y_b + F_{bb}\right) \bar{B}_{a\mid b} + F_{ab} F_{bb}^{-1} F_{ba}\right) \right] \right\rbrace \\
	&= \exp\left\lbrace -\frac{1}{2} \text{tr}\left[ \Sigma_{aa.b}^{-1} \left(\left( B_{a\mid b} - \bar{B}_{a\mid b} \right)^T \left( B_{a\mid b} - \bar{B}_{a\mid b} \right) \right. \right. \right. \\
	& \qquad\qquad\qquad\left.\left.\left. + \left(Y_a - Y_b F_{a\vert b}\right)^T \left(I_n + Y_b F_{bb}^{-1} Y_b^T\right)^{-1} \left(Y_a - Y_b F_{a\vert b}\right) \right) \right] \right\rbrace .
\end{split}
\end{equation*}
Here the last equality is obtained using Proposition~\ref{prop5}. As a result,
\begin{equation*}
\begin{split}
p(Y_a \mid Y_b, \Sigma_{aa.b}) &= \int p(Y_a, B_{a\mid b}\mid Y_b, \Sigma_{aa.b}) dB_{a\mid b}\\
&\propto \exp\left\lbrace -\frac{1}{2} \text{tr}\left[ \Sigma_{aa.b}^{-1} \left(Y_a - Y_b F_{a\vert b}\right)^T \left(I_n + Y_b F_{bb}^{-1} Y_b^T\right)^{-1} \left(Y_a - Y_b F_{a\vert b}\right) \right] \right\rbrace ,
\end{split}
\end{equation*}
and $\text{vec}(Y_a) \mid Y_b, \Sigma_{aa.b} \sim N_{n\times 2}\left(\text{vec}(Y_b F_{a\mid b}), \Sigma_{aa.b} \otimes \left(I_n + Y_b F_{bb}^{-1} Y_b^T\right)\right)$.

Now, by the scaling property of the multivariate Gaussian distribution it is found that
$$\text{vec}\left\lbrace\left(I_n + Y_b F_{bb}^{-1} Y_b^T\right)^{-\frac{1}{2}}\left(Y_a - Y_b F_{a\mid b} \right) \right\rbrace \mid \Sigma_{aa.b} \sim N_{n\times 2}\left(0, \Sigma_{aa.b} \otimes I_n \right),$$
and it follows that \citep[Theorem 3.2.2.]{gupta2000}
$$(Y_a - Y_b F_{a\vert b})^T (I_n + Y_b F_{bb}^{-1} Y_b^T)^{-1} (Y_a - Y_b F_{a\vert b})\sim W_2(\Sigma_{aa.b}, n).$$

\end{proof}

\end{appendices}

\end{document}